\documentclass[journal]{IEEEtran}
\usepackage{cite}
\usepackage{amsmath,amssymb,amsfonts}
\usepackage{amsthm}
\newtheorem{proposition}{Proposition}
\newtheorem{theorem}{Theorem}

\usepackage[linesnumbered, ruled, vlined]{algorithm2e}
\SetAlFnt{\small}
\SetAlgoNlRelativeSize{-1}
\SetKwSty{textbf}
\SetArgSty{textnormal}
\usepackage{graphicx}

\usepackage[svgnames, table]{xcolor}
\definecolor{blue}{RGB}{0,0,255}
\definecolor{red}{RGB}{255,0,0}
\definecolor{green}{RGB}{0,128,0}
\newcommand{\blue}[1]{\textcolor{blue}{#1}}
\newcommand{\red}[1]{\textcolor{red}{#1}}

\usepackage{textcomp}
\usepackage{booktabs}
\usepackage{multirow}
\usepackage{xspace}
\usepackage{mdframed}
\usepackage{placeins}
\usepackage{hyperref}
\usepackage{etoolbox}
\usepackage{bm}

\def\BibTeX{{\rm B\kern-.05em{\sc i\kern-.025em b}\kern-.08em
    T\kern-.1667em\lower.7ex\hbox{E}\kern-.125emX}}

\def\cat{{CAT}\xspace}

\newcommand{\xtxsend}[3]{$\mathsf{Send}(#1,\allowbreak{}#2,\allowbreak{}#3)$}

\def\ie{{i.e.},~}
\def\eg{{e.g.},~}

\def\emptyset{\varnothing}

\def\z3{{\sc Z3}\xspace}

\newcommand{\figref}[1]{Figure~\ref{#1}}
\newcommand{\secref}[1]{Sec.~\ref{#1}}

\def\accepted{\textsf{Accepted}}
\def\postponed{\textsf{Postponed}}
\def\propose{\textsf{Propose}}
\def\status{\textsf{Status}\xspace}
\def\failure{\textsf{failure}\xspace}
\def\success{\textsf{success}\xspace}
\def\pending{\textsf{pending}\xspace}
\def\depth{\textsf{depth}\xspace}

\def\emptyset{\varnothing}

\def\reads{\mathsf{reads}}

\def\writes{\mathsf{writes}}

\def\externalreads{\mathsf{xReads}}
\def\externalwrites{\mathsf{xWrites}}
\def\Next{\mathsf{Next}}
\def\range{\mathsf{Range}}
\def\mostrecentwrite{\mathsf{LastWrite}}

\def\calD{{\cal D}}
\def\calg{{\cal G}}
\def\calI{{\cal I}}
\def\calK{{\cal K}}

\def\calN{{\cal N}}

\def\calT{{\cal T}}
\def\calV{{\cal V}}
\def\Skip{\textsf{Skip}}
\def\dom{\textsf{dom}}
\def\codom{\textsf{codom}}
\def\changeset{\mathsf{ChangeSet}}
\def\trace{\mathsf{trace}}
\def\memtrace{\mathsf{MemTr}}

\newcommand{\directReadDependsOn}[1]{\rightarrow^R_{#1}}
\newcommand{\directWriteDependsOn}[1]{\rightarrow^W_{#1}}
\newcommand{\directDependsOn}[1]{\rightarrow_{#1}}
\newcommand{\dependsOn}[1]{\rightarrow_{#1}^+}

\newcommand{\mutuallyDependsOn}[1]{\leftrightarrow_{#1}^+}
\newcommand{\directIndependent}[1]{\not\leftrightarrow_{#1}}
\newcommand{\independent}[1]{\not\leftrightarrow_{#1}^+}

\def\nonTP{nTP\xspace}

\def\emptymap{[ \,\, ]}

\definecolor{ForestGreen}{RGB}{34,139,34}

\let\oldsigma\sigma
\renewcommand{\sigma}{\oldsigma}

\newtheorem{observation}{Observation}
\newtheorem{lemma}{Lemma}
\newtheorem{definition}{Definition}

\begin{document}

\title{CATs: Secure Blockchain Interoperability with Cross-chain Atomic Transactions}
\author{%
  \begin{tabular}{c@{\hspace{4em}}c}
    Andreas Penzkofer & Franck Cassez \\[2pt]
    Move Industries & Movement Labs \\
    San Francisco, USA & Sydney, Australia
  \end{tabular}%
}

\maketitle

\begin{abstract}
  We propose a protocol for \emph{cross-chain atomic transactions} (\cat{}s), enabling composable atomic execution across different blockchains. The protocol addresses the key interoperability challenge of providing atomicity guarantees in the presence of asynchronous communication and Byzantine actors. It preserves chain autonomy by allowing each blockchain to maintain its own execution model while participating in coordinated cross-chain operations. The design introduces a shared coordination layer involving sequencers, transaction processors, a coordinator, and a confirmation layer which together ensure that either all parts of a \cat succeed or none do. ~\raisebox{-0.75mm}{\includegraphics[height=1em]{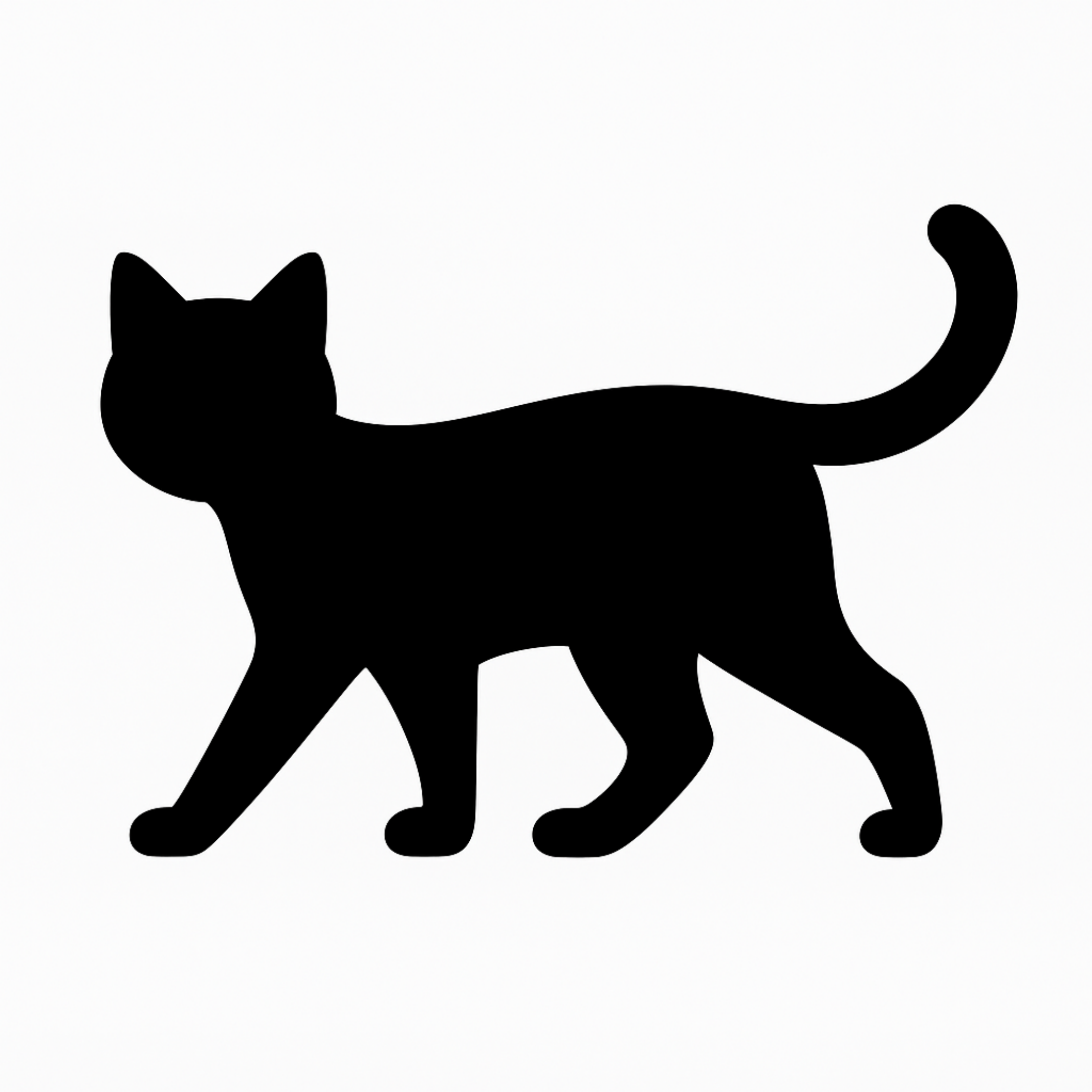}}

  To prevent unnecessary blocking, we separate transaction execution into accepted and postponed sets, with the coordination layer resolving the outcomes of \cat{}s within a few rounds. We further introduce timeouts and dependency-depth bounds for liveness and mitigation of cascading delays. 
  Our formal analysis establishes strong safety and liveness guarantees and demonstrates that the protocol achieves minimal blocking for independent transactions while ensuring bounded blocking time for dependent transactions. Experimental evaluation shows high \cat{} success when cross-chain transactions are a modest share of traffic, and characterizes the \cat{}-lifetime trade-off between success and dependent-transaction latency.

  This protocol enables fast, secure, and deterministic atomic cross-chain execution while preserving chain autonomy, providing a foundation for scalable blockchain interoperability solutions.
\end{abstract}

\section{Introduction}

Blockchain systems are increasingly evolving into multi-chain and multi-layered ecosystems, offering new scalability and specialization benefits at the cost of significant interoperability challenges. 
In a rollup-centric Ethereum landscape, dozens of Layer-2 (L2) networks now operate in parallel, while alternative Layer-1 (L1) chains proliferate. This proliferation has introduced fundamental issues: fragmented liquidity, cross-chain transaction latency, and a lack of atomicity in inter-chain operations.
While transactions may be faster and cheaper on individual chains, the overall user experience and capital efficiency suffer dramatically from this fragmentation.

\subsection{Multi-Chain Interoperability}

In the following, we outline the challenges and opportunities in multi-chain interoperability.
\subsubsection*{Fragmentation of Liquidity Across Chains and Layers}

Scaling via multiple chains or rollups inherently fragments liquidity and state. Each L2 operates as a separate execution environment with its own asset instances and protocols. Consequently, capital that was once pooled on L1 becomes trapped in isolated silos across chains, dramatically reducing market depth and efficiency. 
DeFi protocols struggle to maintain deep liquidity across all networks, forcing users into suboptimal pricing or complex multi-step maneuvers.
A simple token purchase may require network switching, asset bridging, and fee payments across multiple chains, creating significant friction and risk.
This \emph{maze of isolated chains} fundamentally undermines the composability and seamlessness expected from a unified platform.
Fragmented liquidity has emerged as a critical scaling trade-off, where throughput gains and specialized execution come at the expense of capital efficiency and user convenience.
To address this, projects like AggLayer~\cite{agglayer2025} propose a neutral cross-chain execution layer that unifies liquidity, users, and state across sovereign chains, using Ethereum as the base for finality~\cite{polygon2024misconceptions,espresso2024agglayer}. 
By introducing a shared execution layer and unified bridge model, such designs aim to concentrate liquidity and support seamless cross-chain interactions without manual bridging.
The Ethereum community also explores standards, which allow applications to operate across multiple chains through a single abstracted interface. ERC-7683 in particular defines a way for users to express an intent (e.g. swap assets across chains) as one high-level request, which a network of relayers then executes across chains on the user's behalf~\cite{erc7683}. 
This and other standards \cite{Paradigm2023Intents, AnomaIntents2022} seek to give users the impression of a unified liquidity pool and one-step access to all chains. 
The solutions recognize that a multi-chain ecosystem needs to restore composability and liquidity of a single network, most prominently through a common coordination layer.

\subsubsection*{Latency in Cross-Chain Transactions}

Cross-chain transactions can introduce significant latency, which hampers user experience and protocol composability. Unlike single-chain operations, they involve multi-step processes—such as locking assets, relaying messages, and unlocking or minting on the target chain—with added confirmation delays for security. For example, moving assets from Ethereum to a rollup can take minutes, while withdrawals from optimistic rollups may take days due to finality or challenge periods. This latency prevents atomic execution of multi-chain DeFi strategies (for instance, collateralizing an asset on one chain to borrow on another, or arbitraging price differences across exchanges on two chains), forcing developers to break interactions into separate steps vulnerable to market changes. As a result, users avoid cross-chain workflows, and composability suffers. Emerging solutions like Espresso aim to address this by enabling fast interim confirmations and near-instant cross-rollup communication without waiting for full L1 settlement~\cite{espresso2024agglayer}. Such protocols rely on a confirmation layer to provide fast ``finality'', often utilizing a separate consensus mechanism.

\subsubsection*{Lack of Atomicity in Multi-Chain Interactions}

Cross-chain operations typically lack atomicity, meaning that a single logical transaction spanning multiple blockchains cannot be executed in an all-or-nothing fashion. Unlike single-chain transactions, there is no shared view or transaction manager to guarantee coordinated commits. 
This creates failure scenarios where one part of a trade might succeed while another fails, leaving users in inconsistent or loss-bearing states. Even protocols like hash-time-lock contracts (HTLCs), designed to enforce mutual exchange or refund, rely on timeouts and can fail under delays or crashes, undermining atomic guarantees~\cite{zakhary2020atomic}. 
The absence of atomicity complicates dApp design, requiring developers to simulate rollback mechanisms across independent chains, which introduces security risks and implementation complexity.
Research proposals such as decentralized witness networks and cryptographic commit protocols offer potential solutions~\cite{zakhary2020atomic,Cai2024cats-through-state-layers}, but they are not yet widely adopted. For now, the inability to ensure atomic multi-chain interactions remains a major barrier to seamless cross-chain composability.

\subsubsection*{Optimistic Concurrency Control}

OCC~\cite{kung1981optimistic} is a concurrency paradigm in transactional systems that assumes conflicts are rare and allows transactions to execute speculatively without locks.
Transactions proceed independently and are validated only at commit time, where conflicting ones are aborted and retried. 
This approach has been adopted in high-performance blockchain systems such as Aptos, where the Block-STM engine~\cite{gelashvili2022blockstmscalingblockchainexecution} leverages OCC to parallelize simulation of execution before commitment while ensuring determinism. 
Our protocol shares conceptual similarities—such as partially ordered execution and deferred resolution of outcomes. It does not rely on speculative rollback, but takes into account all relevant state transitions.
It also embeds causal dependencies directly into the scheduling and confirmation layers, using explicit postponement or skipping rules and synchronization guarantees to ensure atomicity, liveness, and efficient composability across chains. The similarity is not coincidental, and \eg ~\cite{Cai2024cats-through-state-layers} implements a cross-chain atomic transaction protocol based on OCC, for instance.

\subsection{Contributions and Paper Structure}

Fragmented liquidity, high latency, and non-atomic operations limit blockchain interoperability. As can be seen in the related work \secref{sec:related-work}, solutions like AggLayer and Espresso aim to bridge these gaps in Ethereum, while other L1s, such as Polkadot and Cosmos, integrate interoperability more natively, albeit at the cost of other constraints. A path forward will blend infrastructure coordination, cryptographic assurance, and protocol design to achieve scalable, seamless multi-chain ecosystems.

In this paper, we present a novel protocol for atomic multi-chain execution that overcomes the limitations of existing interoperability solutions by embedding causal dependencies directly into the scheduling and confirmation process. Our design enables atomic, low-latency, and composable cross-chain interactions without requiring speculative rollback or relying fully on centralized coordination. The protocol ensures correctness through explicit dependency tracking and synchronization guarantees, while maintaining liveness under chain-level asynchrony and partial failures. We detail the architecture, formalize the protocol's guarantees, and evaluate its performance under realistic network conditions and adversarial scenarios. Through simulation, we show that the protocol is performant and practical even under high load and a great degree of contention on state key accesses. 

The remainder of the paper is structured as follows: \secref{sec:motivation} introduces an exemplary motivation for the protocol; \secref{sec:model} introduces the system model and assumptions; \secref{sec:protocol} describes the protocol in detail; \secref{sec:analysis-protocol-correctness} analyzes its correctness and liveness properties; \secref{sec:evaluation} presents an empirical evaluation; and \secref{sec:related-work} discusses related work and key distinctions. We conclude in \secref{sec:conclusion}.


\section{Motivation}
\label{sec:motivation}


To illustrate the need for cross-chain transaction protocols that ensure \textbf{atomicity}, are \textbf{minimally blocking}, and guarantee \textbf{liveness}, we walk through a concrete scenario with two chains, Chain~1 and Chain~2, and a unique asset (token) that exists with a corresponding representation on both chains.

\subsection{Regular transactions.}
Let us assume, for simplicity, that each chain offers a single type of transaction \xtxsend{sender}{recipient}{amount}, where \emph{amount} is the amount of the asset to be transferred from \emph{sender} to \emph{recipient}.
The \emph{state} of the chain is a function that maps each account to its balance.
The execution of a transaction on a chain either succeeds or fails, and the status of the execution of a transaction is either \success or \failure.
When a transaction is successful, the state of the chain is updated accordingly. If it fails the state of the chain is left unchanged.
The status of the execution of a transaction depends on the \emph{state} of the chain.
For instance, \xtxsend{Bob}{Charlie}{200} succeeds from a state $s$ if and only if $s(\textit{Bob}) \geq 200$ \ie Bob's balance is at least $200$. If it succeeds, Bob's balance is decreased by $200$ and Charlie's balance is increased by $200$.
If Bob's balance is less than $200$ in state $s$, the execution of the transaction fails.

\begin{figure}[thbp]
  \centering
  \includegraphics[width=0.49\textwidth]{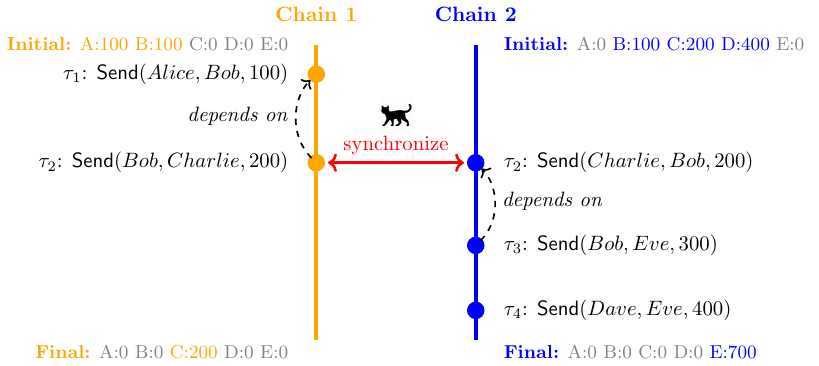}
  \caption{\small A scenario with 5 participants and a cross-chain atomic transaction (CAT) that swaps assets between Bob and Charlie on two chains. Time evolves from top to bottom. Examples of balances are shown for an initial state and after the transactions are executed, for the case where all transactions succeed. Dependencies are shown with dashed lines.}
  \label{fig:example-sync}
\end{figure}

\subsection{Multi-chain transactions.}
For simplicity, assume we have a \emph{shared sequencer} that can order and schedule pairs of transactions on the two chains. We use the pair [t, \Skip] for a transaction t on Chain~1 and no counterparty on Chain~2.
Users submit pairs of transactions (which may contain \Skip) directly to the sequencer, and the sequencer will order the tuple of transactions to be executed on the two chains.
We consider a scenario involving several participants: Alice, Bob, Charlie, Dave, and Eve, see \figref{fig:example-sync}. The sequence of transactions to execute is shown in Table~\ref{tab:transactions}.

\begin{table}[htbp]
  \centering
  \begin{tabular}{@{}l@{\hspace{1.5em}}l@{\hspace{2em}}l@{\hspace{0.5em}}l@{}}
    \textbf{Tx} & \textbf{Type} & \textbf{Chain 1} & \textbf{Chain 2} \\
    \midrule
    $\tau_1$ & Regular & \xtxsend{Alice}{Bob}{100} & \Skip \\
    $\tau_2$ & Cross-chain & \xtxsend{Bob}{Charlie}{200} & \xtxsend{Charlie}{Bob}{200} \\
    $\tau_3$ & Regular & \Skip & \xtxsend{Bob}{Alice}{300} \\
    $\tau_4$ & Regular & \Skip & \xtxsend{Dave}{Eve}{400} \\
    \midrule
  \end{tabular}
  \caption{Sequence of transactions to execute.}
  \label{tab:transactions}
\end{table}

The semantics of CATs is that either all parts of a CAT succeed or none do~\cite{lu2024atomicity,Robinson2020Performance,RobinsonRamesh2020Layer2}.
As a result, the status of $\tau_2$ depends on the status of the two sides \xtxsend{Bob}{Charlie}{200} and \xtxsend{Charlie}{Bob}{200} on their respective chains.
Without coordination between the two chains, there are two possible classes of outcomes for each side of $\tau_2$:
on Chain~1, \xtxsend{Bob}{Charlie}{200} succeeds or fails, and on Chain~2, \xtxsend{Charlie}{Bob}{200} succeeds or fails.

This is the source of several difficulties in implementing cross-chain atomic transactions:
\begin{enumerate}
  \item \textbf{No trusted interchain communication.} The two chains are independent and cannot \emph{directly} communicate their states and transactions' statuses.
  \item \textbf{Rollback complexity.} For a cross-chain transaction $\tau_2 = [t_1, t_2]$, once $t_1$ is confirmed on Chain 1, but $t_2$ fails on Chain 2, rolling back $t_1$ is non-trivial (especially in chains where transactions are supposed to be final once included in a block) and may lead to loss of user funds\footnote{For example in the example in Table~\ref{tab:transactions}, if the first part of $\tau_2$ succeeds, but the second part fails, Bob would lose 200 tokens.}.
  \item \textbf{Independent simulation.} Our protocol addresses this through simulation and coordination before commitment~\cite{RobinsonRamesh2020Layer2}. More specifically, we have to first \emph{determine} the status of a transaction on each chain (i.e., simulate it) but cannot commit it yet. 
  \item \textbf{Necessity of a shared confirmation layer.} Next, a consensus has to be reached on the combined outcome. If both sides of a cross-chain transaction succeed, the results can be committed. 
  If at least one side of a cross-chain transaction fails, the chains should agree on the failure and skip it. Without a shared confirmation layer, parts of a CAT may be committed on one chain before the other part is executed on the other chain, or vice versa.  This motivates the need for a common confirmation layer that coordinates execution across chains.
\end{enumerate}

\begin{mdframed}[backgroundcolor=gray!10, linewidth=0.5pt]
  Our protocol must ensure \textbf{atomicity}. This means that either all constituent transactions of a CAT succeed, or all fail. No partial execution is allowed, maintaining consistency across all involved chains.
\end{mdframed}

\medskip


\subsection{Minimally-blocking protocol.}

Some transactions depend on the outcome of other transactions. As shown in BlockSTM~\cite{gelashvili2022blockstmscalingblockchainexecution}, even in the case where execution could be optimistically parallelized, the application of a transaction (\ie its change set) to the state has to await other transactions it depends on. For example, in \secref{sec:motivation}, $\tau_3$ (\xtxsend{Bob}{Eve}{300}) depends on the outcome of $\tau_2$ (\xtxsend{Charlie}{Bob}{200}). 
We say a transaction is \emph{blocking} if it prevents other transactions from being executed until its execution outcome is known.

One way to implement a protocol for \cat{}s is, in principle, to \emph{synchronize} at each \cat{} and block entirely until resolved. 
However, this approach is not optimal as it introduces unnecessary delays. A better approach is to \emph{minimize} the blocking of transactions and allow transactions to be executed (and accepted) out of order, as long as they are independent transactions, see \secref{sec:naive-cat-protocol} for a description of the transaction flow.

For example, consider the scenario in \figref{fig:example-sync} where $\tau_4$ could be blocked because of $\tau_3$. This is not optimal as this transaction is independent of $\tau_2$ and $\tau_3$. 
Maintaining the exact order of transactions is blocking and may delay every other transaction ordered after a \cat.
  For instance, the transactions $\tau_3$ and $\tau_4$ can only be processed after $\tau_2$ has been completed.
  This may be unavoidable for some transactions like $\tau_3$ as its execution depends on the outcome of $\tau_2$.
  However, $\tau_4$ does not depend on the result of $\tau_2$.

  \begin{mdframed}[backgroundcolor=gray!10, linewidth=0.5pt]
    The \cat{} protocol should be \textbf{minimally-blocking}. This means that transactions that are independent of blocking transactions should be allowed to be executed without any delay.
  \end{mdframed}

\subsection{Eventual progress.}
  
The chains have to synchronize (and agree) on transaction outcomes that affect multiple chains. In reality, the two chains may run at different speeds, and one chain may have to wait for the other before committing its new state.
This synchronization requirement can lead to significant delays, as faster chains must wait for slower ones to complete their operations. For this reason, we must ensure liveness of the protocol, which can be achieved by introducing a timeout mechanism and using a shared confirmation layer.

\begin{mdframed}[backgroundcolor=gray!10, linewidth=0.5pt]
 The \cat{} protocol must ensure \textbf{liveness}. This means all transactions eventually reach a final status (success or failure) within a bounded time, preventing indefinite blocking and ensuring the system can make progress even in the presence of failures or extreme network delays.
\end{mdframed}


\section{Transaction and Computation Model}\label{sec:model}

In this section, we formalize the semantics of transactions, their computation, and the dependencies between them. We define the mathematical foundation for transaction execution, state transitions, and dependency relationships that will be used throughout the paper. We assume homogeneous chains sharing the same VM architecture; the Transaction Processor (TP) is a native component of each chain, not an external add-on. Heterogeneous chain integration is feasible in principle given appropriate VM-level adaptations, but is outside the scope of this work.
While we use cross-chain atomic transactions (\cat{}s) as examples, the focus here is on the transaction model itself rather than the multi-chain aspects. The multi-chain coordination and atomicity properties are addressed in Section~\ref{sec:protocol}.

\subsection{Maps and sequences}
\label{sec:notations}

A \emph{partial map} is a partial function $m: \calK \rightharpoonup \calV$ that may not be defined on all of $\calK$.
The \emph{domain} of $m$, $\dom(m) \subseteq \calK$, is the set of \emph{keys} on which $m$ is defined.
The \emph{codomain} of $m$, $\codom(m) \subseteq \calV$, is the set of \emph{values} that $m$ can take.
We write $[k_1 \mapsto v_1, k_2 \mapsto v_2, \cdots, k_n \mapsto v_n]$ for the partial map that maps $k_i$ to $v_i$ for $1 \leq i \leq n$, where $n$ is the total number of key-value pairs in the map.
The \emph{empty map} $\emptymap$ is the partial map with $\dom(\emptymap) = \emptyset$, \ie defined nowhere.
A \emph{total map} is a partial map defined everywhere \ie with $\dom(m) = \calK$.
We write $m: \calK \rightarrow \calV$ ($\rightarrow$ instead of $\rightharpoonup$) for a total map.

Given two partial maps $m_1: \calK \rightharpoonup \calV$ and $m_2: \calK \rightharpoonup \calV$, the \emph{update} of $m_1$ with $m_2$ denoted $m_1 \oplus m_2$ is a partial map defined on $\dom(m_1) \cup \dom(m_2)$ as:
\[
  (m_1 \oplus m_2)(k) =
  \begin{cases}
    m_2(k) & \text{if } k \in \dom(m_2) \\
    m_1(k) & \text{otherwise} \\
  \end{cases}
\]

\medskip

For a sequence $\sigma = e_1.e_2.\cdots.e_{n}$, we let $|\sigma| = n$ be the \emph{length} of the sequence.
The \emph{empty sequence} is $\varepsilon$ and has length $|\varepsilon| = 0$.
We denote $\sigma_1.\sigma_2$ as the concatenation of sequences.
For $1 \leq i \leq |\sigma|$, we write $\sigma[i]$ for $e_i$.
For $1 \leq i \leq j \leq |\sigma|$, $\sigma[i..j]$ is the sequence $e_i.\cdots.e_{j}$.

\subsection{Transactions}
\label{sec:skipping-transactions}

We assume both chains offer the set $\calT$ of transactions.
We require that $\calT$ contains a special transaction $\Skip \in \calT$ that has no effect.
A \cat is a pair $(t_1, t_2)$ with $t_j \in \calT, j \in \{1,2\}$, where $t_1 \neq \Skip$ and $t_2 \neq \Skip$.
If $t_1 = \Skip$ or $t_2 = \Skip$, the transaction is a \emph{regular} transaction (involving only one chain).
This definition enables us to use the same notation for regular transactions and \cat{}s.
Transactions, regular or \cat, are submitted to the network $\calN$ as pairs of transactions, one for each chain.

\subsection{States, change sets, and state transitions}
\label{sec:states-and-change-sets}

A \emph{state} $s$ of chain $C_i, i = 1, 2$ is a total map $s: \calK \to \calV$.
The set $\calK$ can be thought of as the set of memory locations or global storage locations of chain $C_i$, and the set $\calV$ as the set of values that can be stored in these locations.
We assume that both chains have the same type of memory locations $\calK$ and the same type of values $\calV$.
The initial state of the chain is the \emph{genesis state} $\calg(C_i)$.

\medskip

A \emph{change set} is a partial map from $\calK$ to $\calV$.
The semantics of the execution of a transaction $t \in \calT$ from a state $s$ is defined by a change set function $\changeset(s, t) : \calK \rightharpoonup \calV$, and we require that:
\[
  \changeset(s, \Skip)=\emptymap
\]
The function $\Next$ defines the state obtained after executing $t$ from $s$:
\[
  \Next(s, t) = s \oplus \changeset(s, t)
\]
We extend $\Next$ to sequences of transactions:
\[
  \Next(s, \varepsilon) =  s \\
  \Next(s, t.\sigma) = \Next(\Next(s, t), \sigma)
\]
To capture the status of the execution of a transaction, we define the \emph{state transition} function $s \xrightarrow{\ t/\alpha \ } s'$ as follows: $\alpha \in \{\textsf{success}, \textsf{failure}\}$ is the status of the execution of transaction $t$ from state $s$, and
$s' = \Next(s, t)$.
The transition function must satisfy the following requirements:
\[
  s \xrightarrow{\ t/\failure \ } s' \implies \changeset(s, t) = \emptymap
\]
\[
  s \xrightarrow{\ \Skip/\alpha \ } s' \implies \alpha = \success \wedge \changeset(s, \Skip) = \emptymap
\]
A consequence of the above is that $s \xrightarrow{\ t/\failure \ } s'$ implies $\Next(s, t) = s$,
and that $\Next(s, \Skip) = s$.

Given a sequence of transactions $\sigma = t_1.t_2.\cdots.t_{n}$, the \emph{execution} $\gamma$ of the sequence is a sequence of state transitions:
\[
  \gamma: s_0  \xrightarrow{\ t_1/\alpha_1 \ } s_1 \xrightarrow{\ t_2/\alpha_2 \ } \cdots \xrightarrow{\ t_{n}/\alpha_{n} \ }
  s_{n}
\]
where $\alpha_i, 1 \leq i \leq n$, is the status of the execution of transaction $t_i$.
The \emph{trace} of $\gamma$ is:
\[
  \trace(\gamma) = t_1/\alpha_1.t_2/\alpha_2.\cdots.t_{n}/\alpha_{n}
\]

\subsection{Reads and writes, memory traces}
\label{sec:reads-and-writes}

The execution of a transaction $t$ from a given state $s$ may read or write a set of memory locations in $\calK$.
We define the following operations: $\reads(k)$ to denote a read operation that reads the value of location $k\in \calK$ and $\writes(k, v)$ to denote a write operation that writes the value $v\in \calV$ to location $k$.

There are several other operations performed during the computation of a transaction (\eg arithmetic, control flow, etc.), but we do not need to consider them for the purpose of this discussion.

The memory \emph{trace}, $\memtrace(s, t)$, resulting from the execution of a transaction $t$ from a state $s$ is a sequence of $\reads(\cdot)$ and $\writes(\cdot,\cdot)$ operations.

For instance, assuming $\calK$ and $\calV$ are the set of non-negative integers, the sequences of operations $\mu_i, 1 \leq i \leq 6$ in Table~\ref{tab:transaction-set} each represent a memory trace. The memory access pattern for these transactions is visualized in \figref{fig:memory-locations}.

\begin{table}[htbp]
  \centering
  \renewcommand{\arraystretch}{1.3}
  \begin{tabular}{clc}
  \arrayrulecolor{gray!30}
  \hline
  \rowcolor{gray!20}
  \textbf{Transaction} & \textbf{Memory Trace} & \textbf{Range} \\
  \hline
  $t_1 = \sigma[1]$ & $\mu_1 = \reads({\color{blue}10}).\writes({\color{red}20}, v_1).\reads({\color{blue}30})$ & $[1,3]$ \\
  \hline
  $t_2 = \sigma[2]$ & $\mu_2 = \reads({\color{blue}40}).\writes({\color{red}50}, v_2).\reads({\color{blue}60})$ & $[4,6]$ \\
  \hline
  $t_3 = \sigma[3]$ & $\mu_3 = \reads({\color{blue}20}).\writes({\color{red}60}, v_3)$ & $[7,8]$ \\
  \hline
  $t_4 = \sigma[4]$ & $\mu_4 = \writes({\color{red}60}, v_4).\reads({\color{blue}50})$ & $[9,10]$ \\
  \hline
  $t_5 = \sigma[5]$ & $\mu_5 = \writes({\color{red}70}, v_5)$ & $[11,11]$ \\
  \hline
  $t_6 = \sigma[6]$ & $\mu_6 = \writes({\color{red}80}, v_6)$ & $[12,12]$ \\
  \hline
  \end{tabular}
  \vspace{4pt}
  \caption{Reference transaction set used throughout the paper. Colored numbers indicate the memory locations of {\color{blue}read} ({\color{red}write}) operations. $v_i$ is the value written to the memory location. For a more visual representation, see Figure~\ref{fig:memory-locations} and Figure~\ref{fig:dependencies}.}
  \label{tab:transaction-set}
\end{table}

\begin{figure}[htbp]
  \centering
  \includegraphics[width=0.48\textwidth]{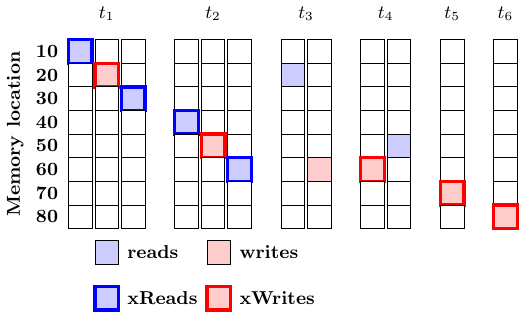}
  \caption{Memory locations matrix showing read and write operations for each transaction from Table~\ref{tab:transaction-set}. Each row represents a memory location, and each column represents an operation in the memory trace of the transaction. External reads and writes for the sequence $\sigma=t_1.t_2.t_3.t_4.t_5.t_6$ are highlighted with borders.}
  \label{fig:memory-locations}
\end{figure}

Given a state $s$ and a sequence of transactions $\sigma$, $\memtrace(s, \sigma)$ is the concatenation of sequences of $\reads$ and $\writes$ to locations in $\calK$ over each transaction $\sigma[i], 1 \leq i \leq |\sigma|$.
We extend $\memtrace$ to sequences of transactions:
\begin{align*}
  \memtrace(s, \varepsilon)&  = \varepsilon &  \\
  \memtrace(s, t.\sigma) & = \memtrace(s, t).\memtrace(\Next(s, t), \sigma) &
\end{align*}

\subsection{Memory ranges}
\label{sec:memory-ranges}

The \emph{range} of each transaction in $\memtrace(s, \sigma)$ is the section of reads and writes occurring when executing the transaction.
We can define the function $\range(s, \sigma, i)$ that returns a (convex) interval of indices (range) of $\memtrace(s, \sigma)$ that corresponds to the execution of the transaction $\sigma[i]$.
For instance, assume the memory trace of the sequence of transactions $\sigma=\sigma[1].\sigma[2]$ (see Table~\ref{tab:transaction-set}) from state $s$ is
\begin{align*}
  \memtrace(s, \sigma) & = \memtrace(s, \sigma[1].\sigma[2]) & \\
  & =  \underbrace{\memtrace(s, \sigma[1])}_{\mu_1}.\underbrace{\memtrace(\Next(s, \sigma[1]), \sigma[2])}_{\mu_2} & \\
  & = \mu_1.\mu_2  &
\end{align*}
The range of $\sigma[1]$ in $\memtrace(s, \sigma)$ is $\range(s, \sigma, 1) = [1,3]$ and the range of $\sigma[2]$ is $\range(s, \sigma, 2) = [4,6]$ as shown in~\figref{fig:memtrace}.

\begin{figure}[htbp]
  \centering
  \includegraphics[width=0.49\textwidth]{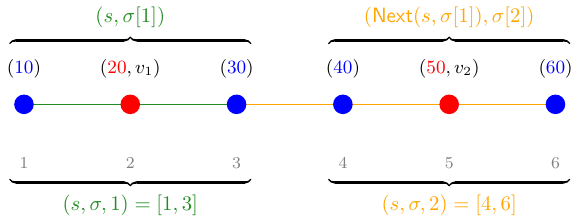}

  \caption{Memory traces $\color{ForestGreen}\mu_1$ and $\color{orange}\mu_2$ with ranges, as shown in Figure~\ref{fig:memory-locations} and Table~\ref{tab:transaction-set}. {\color{blue}Read} ({\color{red}Write}) operations are shown with their value written to the memory location.}
  \label{fig:memtrace}
\end{figure}

\subsection{External reads and writes}
\label{sec:external-reads-and-writes}

Given a memory trace $\mu$, we define the \emph{external reads} $\externalreads(\mu)$ as the set of memory locations that are read but not preceded by a write to the same location within the trace. These definitions use the indexing structure of memory traces as formalized in the previous section on memory ranges:
\begin{align}
  \externalreads(\mu) = \{ k \in \calK \mid \exists i \in [1, |\mu|] : \mu[i] = \reads(k) \wedge \nonumber \\
  \forall j < i : \mu[j] \neq \writes(k, \cdot) \} \nonumber
\end{align}

Similarly, we define the \emph{external writes} $\externalwrites(\mu)$ as a partial map that captures the final key-value pair written to each memory location in the trace:
\begin{align}
  \externalwrites(\mu) = [k \mapsto v \mid \exists i \in [1, |\mu|] : \mu[i] = \writes(k, v) \wedge \nonumber \\
  \forall j > i : \mu[j] \neq \writes(k, \cdot)] \nonumber
\end{align}

Let $\mu' = \mu_1.\mu_2.\mu_3.\mu_4.\mu_5.\mu_6$ be the concatenated memory trace of the transactions from Table~\ref{tab:transaction-set}.
The external reads are $\externalreads(\mu') = \{ 10, 30, 40, 60 \}$ because $\reads(20)$ in $t_3$ is preceded by a write to $20$ (in $t_1$) and $\reads(50)$ in $t_4$ is preceded by a write to $50$ (in $t_2$), so neither are external reads.
We have $\externalwrites(\mu') = \{ 20 \mapsto v_1, 50 \mapsto v_2, 60 \mapsto v_4, 70 \mapsto v_5, 80 \mapsto v_6 \}$, since $t_4$ writes to $60$ again after $t_3$; see Figure~\ref{fig:memory-locations}.

\subsection{Dependencies}
\label{sec:dependencies}

Given a sequence of transactions $\sigma$ and a range index $0 \leq i < |\memtrace(\sigma)|$ such that there is a read at this location \ie $\memtrace(\sigma)[i] = \reads(\ell)$ for some $\ell \in \calK$, we define $\mostrecentwrite(\sigma, i)$ to be the most recent write to this location in the memory trace.
Let $K = \{ j < i \mid \memtrace(\sigma)[j] = \writes(\ell, -) \}$. The function $\mostrecentwrite(\sigma, i)$ is defined\footnote{With the convention that $\max \varnothing = - \infty$.} by $max_{k \in K} k$.

Given a state $s$, the result of the execution of $\sigma$ from $s$ is given by the $\Next$ function as defined in Section~\ref{sec:states-and-change-sets}.
This assumes that we execute the transactions in sequence. However, for performance reasons, we may want to compute some of them in parallel (as is the case with BlockSTM), or as we will later see in \secref{sec:naive-cat-protocol}, we may want to execute some of them out of order.
For the latter to be possible, we need to make sure that transactions executed out of order give the same results as if they were executed in sequence.
This can be formalized by a notion of \emph{dependency} between transactions.

\subsubsection{Read-Write dependency.}

A transaction $\sigma[j]$ \emph{depends directly on} a transaction $\sigma[i]$ if $i < j$ and $\sigma[j]$ reads from a location that was written by $\sigma[i]$. This is a \emph{read-write dependency}.
Given a state $s$ and $i < j$, 
we define the dependency relation $\cdot \directReadDependsOn{s} \cdot$ between transactions in a sequence of transactions $\sigma$ by: let $\mu = \memtrace(s, \sigma)$,
\[
\begin{gathered}
\sigma[j] \directReadDependsOn{s} \sigma[i] \iff \\[2pt]
\begin{cases}
  \exists k \in \calK, & \text{[a memory location]} \\
  \exists n \in \range(s,\sigma,j), & \text{[an operation in $\sigma[j]$]} \\
  \mu[n] = \reads(k), & \text{[which is a $\reads(k)$]} \\
  \mostrecentwrite(\mu,n) \in \range(s,\sigma,i), &
     \text{[Last write is by $\sigma[i]$]}
\end{cases}
\end{gathered}
\]

\subsubsection{Write-Write dependency.}

A transaction $\sigma[j]$ has a \emph{write-write dependency} on transaction $\sigma[i]$ if $i < j$ and both transactions write to the same memory location. We define this as:
\[
\begin{gathered}
  \sigma[j] \directWriteDependsOn{s} \sigma[i] \iff \\[2pt]
  \begin{cases}
    \exists k \in \calK, & \text{[a memory location]} \\
    \exists n \in \range(s, \sigma, j), & \text{[an operation in $\sigma[j]$]} \\
    \exists m \in \range(s, \sigma, i), & \text{[an operation in $\sigma[i]$]} \\
    \mu[n] = \writes(k, \cdot), & \text{[which is a $\writes(k, \cdot)$]} \\
    \mu[m] = \writes(k, \cdot), & \text{[and $\sigma[i]$ also writes to $k$]}
  \end{cases}
\end{gathered}
\]
A write-write dependency is, for example, relevant if we consider the final value written to a location.

Note that nonce ordering is itself a dependency: nonce $n{+}1$ from the same account depends on nonce $n$, which is a special case of a write-write dependency on the account's nonce storage slot. Monotonicity is thus enforced naturally by the dependency mechanism.

\subsubsection{dependsOn.}

We can define the \emph{dependsOn} relation as the combination of the direct read-write dependency and the direct write-write dependency.
\begin{align*}
  \sigma[j] \directDependsOn{s} \sigma[i] \iff
  \begin{cases}
    \sigma[j] \directReadDependsOn{s} \sigma[i] \wedge \sigma[j] \directWriteDependsOn{s} \sigma[i]
  \end{cases}
\end{align*}

We can define the transitive closure, \emph{dependsOn}, $\cdot \dependsOn{s} \cdot$ of the direct dependency relation as the smallest relation that satisfies:
\begin{align*}
  \sigma[j] \directDependsOn{s} \sigma[k] \implies \sigma[j] \dependsOn{s} \sigma[k] & \\
  \sigma[j] \dependsOn{s} \sigma[k] \wedge \sigma[k] \directDependsOn{s} \sigma[i] \implies \sigma[j] \dependsOn{s} \sigma[i] &
\end{align*}
A transaction $\sigma[j]$ \emph{depends on} a transaction $\sigma[i]$ if $\sigma[j] \dependsOn{s} \sigma[i]$; otherwise, we say that $\sigma[j]$ \emph{does not depend on} $\sigma[i]$ and write $\sigma[j] \not\dependsOn{s} \sigma[i]$.

\subsubsection{Independence.}
We define the \emph{mutual dependency} relation as the case where two transactions depend on each other in either direction:
\begin{align*}
  \sigma[i] \mutuallyDependsOn{s} \sigma[j] \iff
  \sigma[i] \dependsOn{s} \sigma[j] \ \vee\  \sigma[j] \dependsOn{s} \sigma[i]
\end{align*}
Similarly, we say that two transactions $\sigma[i]$ and $\sigma[j]$ are \emph{independent} if they are not mutually dependent, i.e.,
\begin{align*}
  \sigma[i] \independent{s} \sigma[j] \iff
  \sigma[i] \not\dependsOn{s} \sigma[j] \ \wedge\  \sigma[j] \not\dependsOn{s} \sigma[i]
\end{align*}

\begin{lemma}[Independence and disjoint memory access]
\label{lem:independence-disjoint-access}
Two transactions $\sigma[i]$ and $\sigma[j]$ are directly independent iff their external reads and writes of $\sigma[i]$ are completely disjoint from the external reads and writes of $\sigma[j]$, \ie:
\[
  \begin{gathered}
  \sigma[i] \directIndependent{s} \sigma[j] \iff \\[2pt]
    {\setlength{\arraycolsep}{1pt} 
    \begin{array}{ccc}
      \left(\externalreads(\memtrace(s, \sigma[i])) \cup \externalwrites(\memtrace(s, \sigma[i]))\right) & \cap & \\
      \left(\externalreads(\memtrace(s, \sigma[j])) \cup \externalwrites(\memtrace(s, \sigma[j]))\right) & = & \emptyset
    \end{array}}
  \end{gathered}
\]
\end{lemma}
\begin{proof}
By definition, $\sigma[j] \directIndependent{s} \sigma[i]$ means $\sigma[j]$ does not read or write to any location written by $\sigma[i]$. For direct independence, both $\sigma[j] \directIndependent{s} \sigma[i]$ and $\sigma[i] \directIndependent{s} \sigma[j]$ must hold, which requires that external reads and writes are completely disjoint in both directions.

Conversely, if external reads and writes are completely disjoint, then $\sigma[j] \not\directReadDependsOn{s} \sigma[i]$, $\sigma[i] \not\directReadDependsOn{s} \sigma[j]$, $\sigma[j] \not\directWriteDependsOn{s} \sigma[i]$, and $\sigma[i] \not\directWriteDependsOn{s} \sigma[j]$, establishing direct independence. 
\end{proof}

Intuitively, we can extend this to sequences. I.e., instead of $\sigma[i]$ being a transaction, we may also replace it by an entire sequence $\sigma'$.

\subsubsection{Dependency graph.}
The graph of the dependency relation in state $s$ for sequence of transactions $\sigma$ is the
\emph{dependency graph} $\calD(s, \sigma)$, a directed acyclic graph (DAG), with set of nodes $\sigma[j], 1 \leq j \leq |\sigma|$, and the set of directed edges defined by the pairs $(\sigma[j], \sigma[i])$ such that $\sigma[j] \directReadDependsOn{s} \sigma[i]$.
\medskip

From the definition of the dependency graph, it follows that $\sigma[j] \dependsOn{s} \sigma[i]$ if and only if there is a path from $\sigma[j]$ to $\sigma[i]$ in the dependency graph $\calD(s, \sigma)$.
For example, consider the sequence of transactions $\sigma' = t_1.t_2.t_3.t_4$, from Table~\ref{tab:transaction-set}.
The DAG is depicted in \figref{fig:dependencies}, left.
From state $s_0$, this sequence of transactions produces the execution:
\[
  s_0 \xrightarrow{\ t_1/\alpha_1\ } s_1 \xrightarrow{\ t_2/\alpha_2\ } s_2 \xrightarrow{\ t_3/\alpha_3\ } s_3 \xrightarrow{\ t_4/\alpha_4\ } s_4
\]
In this example, $t_2$ does not depend on any other transaction and is an end node. However, $t_3 \directDependsOn{s} t_1$ because $t_3$ reads from location $20$ which was last written by $t_1$. Similarly, $t_4 \directDependsOn{s} t_3$ as it writes to location $60$ which was last written by $t_3$. Also, $t_4 \directDependsOn{s} t_2$ as it reads from location $50$ which was last written by $t_2$. 

\begin{figure}[htbp]
  \centering
  \includegraphics[width=0.49\textwidth]{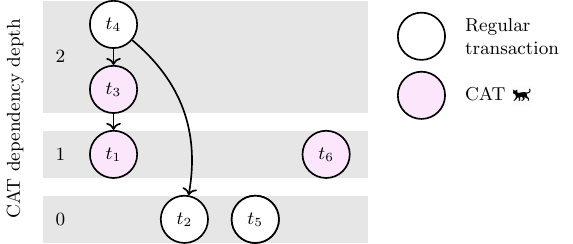}
  \caption{\small Dependencies between transactions in $\sigma$ displayed in a dependency graph $\calD$ if all transactions are successful. If $\sigma[j] \directDependsOn{s} \sigma[i]$ there is an edge from $\sigma[j]$ to $\sigma[i]$. Transactions $t_1$, $t_3$, and $t_6$ are each part of (separate) CATs and increment the CAT dependency depths. For transaction details, see Table~\ref{tab:transaction-set}.}
  \label{fig:dependencies}
\end{figure}

\subsubsection{CAT dependency depth.}
In a cross-chain transaction setting, whether any part of a \cat is committed or not depends on the successful computation of the other parts of the \cat (on the other chains).
For regular transactions on a single chain, this is not a problem, as in the description of our model, the other chain executes a $\Skip$ which is always successful. Hence, the result of the computation can be committed straightaway.\footnote{Unless it depends on a CAT, in which case it is added to a pending list.}
In contrast, for a \cat $T=(t_1, t_2)$ we have to wait before committing the results of the involved transactions.
This means that if there is a \cat in the sequence $\sigma$, we can execute all the regular transactions before the \cat, but the transactions after the \cat may be delayed and postponed to the next block or beyond.
However, if we can determine that some transactions after the \cat are independent of the \cat, we can execute them. This has similarity to BlockSTM, where independent transactions are computed in parallel; however, it goes beyond this by allowing to execute transactions out of order.

\vspace{0.3cm}

Let us identify how we determine whether a transaction is independent of a \cat.
Assume we have on Chain~1 the above CAT $T$, a sequence of regular transactions $\sigma$ and a state $s$.
There are two possible scenarios for $T$: \success or \failure.
We, thus, must consider two sequences of transactions: one in which $T$ succeeds, resulting in sequence $\sigma_\success=t_1.\sigma$, and one in which $T$ fails, resulting in sequence $\sigma_\failure=\Skip.\sigma$.
As one may expect, as the number of CATs increases, this could result in an exponentially growing number of scenarios.
Efficient algorithms for managing multiple conflicting states have been presented, for example, in~\cite{müller2023realitybasedutxoledger}. However, it remains important to limit the number of scenarios. A useful metric that can be utilized to limit the number of branching scenarios is to limit the number of CATs that can be pending and dependent on each other. 

\vspace{0.3cm}

We define the \emph{CAT dependency depth} $\mathsf{depth}_{s,\sigma}$ of a transaction $\sigma[i] \in \sigma$ on $s$ recursively as follows:
Let $D_i = \{\, j < i \mid \sigma[i] \mutuallyDependsOn{s} \sigma[j] \,\}$ be the set of prior transactions in $\sigma$ that $\sigma[i]$ mutually depends on with respect to $s$. Then:
\[
\begin{gathered}
\mathsf{depth}_{s,\sigma}(\sigma[i]) = \\[2pt]
\max\left( \{\, \mathsf{depth}_{s,\sigma}(\sigma[j]) \mid j \in D_i \,\} \cup \{0\} \right) + 
\begin{cases}
1 & \sigma[i] \text{ is CAT}, \\
0 & \text{otherwise}.
\end{cases}
\end{gathered}
\]

\noindent In \figref{fig:dependencies}, we show the CAT dependency depths for the transactions in Table~\ref{tab:transaction-set}.

\subsection{Change Sets and External Writes}
\label{sec:change-sets-external-writes}

In the semantics of transactions, there should be a natural correspondence between the change set of a transaction and the external writes of its memory trace.
The change set of a transaction $t$ describes the updates \emph{after} the transaction has been executed, whereas the external writes of the memory trace of $t$ describe the memory locations that are written to by $t$.
The resulting sets are identical.

\begin{observation}[Change sets are external writes]
  \label{thm:change-sets-external-writes}
  For any state $s$ and a transaction $t$,
  \[
    \changeset(s, t) = \externalwrites(\memtrace(s, t)) 
  \]
\end{observation}

Symmetrically, the semantics of a transaction can only depend on the memory it reads during the execution and is not impacted by memory locations that are not read.
We denote $s \equiv_K s'$ if $\forall k \in K, s'[k] = s[k]$ \ie the states $s$ and $s'$ agree on the values of the keys in $K$.
\begin{observation}[State equivalence and external reads]
  \label{thm:state-equivalence-external-reads}
  Let $s$ be a state and $t$ a transaction.  Let $K = \externalreads(\memtrace(s, t))$.
  For any state $s' \equiv_K s$ we have:
  \begin{equation}
    \memtrace(s', t) = \memtrace(s, t)
  \end{equation}
  Combined with Observation~\ref{thm:change-sets-external-writes}, we have:
  \begin{equation}
    \changeset(s', t) = \changeset(s, t)
  \end{equation}
\end{observation}

\subsection{Independent transactions and change sets}
\label{sec:independent-transactions-change-sets}

If a transaction $\sigma[j] \not\dependsOn{s} \sigma[i]$ (see \secref{sec:dependencies}) with respect to a state $s$, then the change sets of the two transactions can be \emph{computed in parallel}. 
If two transactions are independent (i.e., $\sigma[i] \independent{s} \sigma[j]$, see \secref{sec:dependencies}) with respect to a state $s$, then the change sets of the two transactions can also be \emph{applied out of order}.

\begin{mdframed}[backgroundcolor=gray!10, linewidth=0.5pt]
  In BlockSTM~\cite{Gelashvili2023}, the term \emph{execution of incarnations} is used for what we call \emph{computed}. In literature, execution is frequently considered to be a successful application of a transaction to a state. Hence, we choose to say a transaction is \emph{executed} if it is computed, successfully validated, and applied to the state.
\end{mdframed}

\begin{lemma}[Parallel computation of non-dependent transactions]
  \label{thm:parallel-computation-independent-transactions}
  Let $\sigma = t_1.t_2$ be a sequence of two transactions and $s_0$ a state.
  Assume $t_2 \not\directDependsOn{s_0} t_1$ and the execution of $\sigma$ from $s$ produces the execution $s_0 \xrightarrow{\ t_1\ } s_1 \xrightarrow{\ t_2\ } s_2$.
  Then
  \begin{equation*}
    \changeset(s_1, t_2) = \changeset(s_0, t_2)
  \end{equation*}
\end{lemma}
\begin{proof}
  Let $\mu_1 = \memtrace(s_0, t_1)$ and $\mu_2 = \memtrace(s_1, t_2)$ denote the memory traces of $t_1$ and $t_2$ from the initial state $s_0$ and the state $s_1$, respectively.
  Let $K_2 = \externalreads(\mu_2)$.
  If $K_2 = \varnothing$, then $s_0 \equiv_{K_2} s_1$ and by Observation~\ref{thm:state-equivalence-external-reads}, we have $\changeset(s_0, t_2) = \changeset(s_1, t_2)$.
  Otherwise let $k \in K_2$. Then $k$ cannot be written to in $\mu_1$. The proof is by contradiction. Assume $\writes(k, v) \in \externalwrites(\mu_1)$. Then $t_2 \directDependsOn{s_0} t_1$ which is a contradiction.
  As a result, $\writes(k, v) \not\in \externalwrites(\mu_1)$, by Observation~\ref{thm:change-sets-external-writes} $\writes(k, v) \not\in \changeset(s_0, t_1)$ and, therefore, $s_0[k] = s_1[k]$ by definition of $\Next$ and $s_1 = \Next(s_0, t_1) \oplus \changeset(s_0, t_1)$.
  This implies $s_0 \equiv_{K_2} s_1$, and by Observation~\ref{thm:state-equivalence-external-reads} we have $\changeset(s_0, t_2) = \changeset(s_1, t_2)$. 
\end{proof}

In BlockSTM~\cite{Gelashvili2023}, the change sets of transactions are computed optimistically in parallel. 
Once the change sets are computed, they await to be applied to the source state $s$ to obtain the next state $s'$. 
In order to be applied, they must pass a validation step. 
This validation step ensures that the order of transactions is respected. 
For a transaction for which it is later discovered that it has dependencies on other transactions, the change set is not applied and the transaction has to be re-computed.

In the example of \figref{fig:dependencies}, this would imply that the transactions $t_1$ and $t_2$ can be successfully computed in parallel because they do not depend on each other.
The transaction $t_3$ depends on $t_1$ and must be re-computed (if it would be computed in parallel).
The transaction $t_4$ also depends on both $t_1$ and $t_2$.

We now take this approach one step further. In our model, we permit to apply transactions out of order, i.e., permit to change the order of transactions if the final state remains the same. An application of this, as we will see in \secref{sec:naive-cat-protocol}, is that we can assign transactions to different output streams.

\begin{lemma}[Out-of-order execution of independent transactions]
  \label{lemma:out-of-order-execution-independent-transactions}
  Let $\sigma = t_1.t_2$ be a sequence of two transactions and $s_0$ a state.
  Assume $t_2 \directIndependent{s_0} t_1$ and the execution of $\sigma$ from $s_0$ produces the execution $s_0 \xrightarrow{\ t_1\ } s_1 \xrightarrow{\ t_2\ } s_2$.
  Then
  \begin{equation*}
    \Next(s_0, t_1.t_2) = \Next(s_0, t_2.t_1)
  \end{equation*}
\end{lemma}
\begin{proof}
By definition of independence, we have both $t_2 \not\directDependsOn{s_0} t_1$ and $t_1 \not\directDependsOn{s_0} t_2$.

From Lemma~\ref{thm:parallel-computation-independent-transactions}, since $t_2 \not\directDependsOn{s_0} t_1$, we have $\changeset(s_1, t_2) = \changeset(s_0, t_2)$.
Similarly, since $t_1 \not\directDependsOn{s_0} t_2$, we have $\changeset(s_0, t_1) = \changeset(\Next(s_0, t_2), t_1)$.
Now we can compute both execution orders:

Order 1: $s_0 \xrightarrow{\ t_1\ } s_1 \xrightarrow{\ t_2\ } s_2$
\[
\begin{gathered}
\setlength{\arraycolsep}{1pt}
\begin{array}{cll}
s_2 & = & s_1 \oplus \changeset(s_1, t_2) \\
& = & s_0 \oplus \changeset(s_0, t_1) \oplus \changeset(s_0, t_2)
\end{array}
\end{gathered}
\]

Order 2: $s_0 \xrightarrow{\ t_2\ } s_1' \xrightarrow{\ t_1\ } s_2'$
\[
\begin{gathered}
\setlength{\arraycolsep}{1pt}
\begin{array}{cll}
s_2' & = & s_1' \oplus \changeset(s_1', t_1) \\
& = & s_0 \oplus \changeset(s_0, t_2) \oplus \changeset(s_0, t_1)
\end{array}
\end{gathered}
\]

Since the change sets are the same in both orders and the two ChangeSet operations are commutative, we have $s_2 = s_2'$, which means $\Next(s_0, t_1.t_2) = \Next(s_0, t_2.t_1)$. 
\end{proof}

\medskip

\section{Cross-chain atomic transactions protocol}
\label{sec:protocol}

In this section, we describe the \cat{} protocol. First, we define the requirements for a \cat{} protocol in \secref{sec:requirements-cat}. We then identify the components that are required for a minimal version of the \cat{} protocol in \secref{sec:components}. In \secref{sec:analysis-protocol-correctness}, we address safety and liveness, as well as non-blocking properties.

\subsection{Requirements}
\label{sec:requirements-cat}

We clarify the requirements for a \cat{} protocol to deliver correct behavior for \cat{}s.
This involves identifying the conditions under which transactions should succeed or fail. Our goal is to describe the operational behavior informally and to move toward a precise problem formulation that can guide protocol design and correctness reasoning.

\medskip

\noindent An informal definition of the requirements for a \cat{} protocol is then as follows:
\begin{itemize}
  \item \textbf{Simulation and Proposals.} Each chain reports (\ie proposes) the best effort status of its part of a \cat{}. 
  A best effort status means the chain \textit{simulates} the execution of the transaction for a given state and reports whether the transaction would have \success or \failure, without actually committing any state changes to persistent storage. This allows chains to determine the transaction's status without prematurely modifying their permanent state before consensus is reached. 
  \cite{Cai2024cats-through-state-layers} introduces the term \emph{dirty state layer} to describe the non-permanent storage that can be reverted once the status is determined. We propose a similar data storage layer; however, we require that it supports nesting. Thus, our data storage layer may support several instances—one for each possibility of a \cat{}'s status, see \eg \cite{müller2023realitybasedutxoledger}.
  \item \textbf{Coordination.} The protocol ensures coordination on the status of the \cat{}s. 
  \item \textbf{Confirmation.} The protocol ensures consensus on the success or failure of all \cat{}s:
    \begin{itemize}
      \item \success. A \cat should succeed if and only if all involved chains report successful execution of their respective parts and consensus is reached on the \status (before any potential timeout).
      \item \failure. A \cat should fail in all other cases: if any of the involved chains reports failed execution of their respective parts, or if consensus cannot be reached on the status before the timeout expires, then consensus must be reached on \failure.
    \end{itemize}
\end{itemize}

On a fundamental level, our \cat{} protocol can be viewed in the context of Byzantine State Machine Replication~\cite{lamport1998parttime,castro1999practical}, where the replicated state spans multiple chains and agreement is required on the \success or \failure of coordinated atomic transactions. In this setting, we adopt the following definitions:\footnote{We note that, \eg \cite{Cai2024cats-through-state-layers} uses the terms \emph{authenticity} for verifiability in the context of cross-chain transaction statuses, and \emph{reliability} and \emph{termination} for eventual delivery and commitment of cross-chain transaction statuses (or results more generally). These terms are similar to our definitions, but we use the terms \emph{safety} and \emph{liveness}, as they provide better intuition for the requirements.}

\begin{enumerate}
  \item \textbf{Safety:} If two honest chains decide on the \status of a given \cat{} $T$, they must decide identically. That is, if one chain confirms \success for $T$, no other honest chain may confirm \failure for $T$ or remain undecided.
  \item \textbf{Liveness:} If the confirmation layer operates correctly and a sequencer submits a well-formed \cat{} $T$ to the confirmation layer, then the protocol must eventually decide \success or \failure of $T$ and deliver that decision to the chain.
  \item \textbf{Non-blocking (optional):} Transactions that do not depend on the outcome of a pending \cat{} must be able to proceed without delay. This ensures that CAT-related coordination does not unnecessarily penalize independent execution.
\end{enumerate}

\noindent Cross-chain atomic commitment is also a special case of the well-studied Non-Blocking Atomic Commitment (NBAC) problem~\cite{guerraoui2002nbac}, which addresses crash-tolerant consensus among processes that must all commit or all abort. However, NBAC is defined for the crash fault model and requires a \emph{perfect failure detector}---an oracle that detects crashed participants so that the protocol can abort rather than block indefinitely. Our protocol operates in a Byzantine setting, where such a detector is unavailable; instead, the CL timeout mechanism serves this role: if a chain's executor does not submit its proposal within the CAT lifetime, the coordinator treats it as failed and aborts the CAT. Our requirements map directly to NBAC properties: Safety to \emph{Agreement}, Liveness to \emph{Termination}, the coordinator's commit logic to \emph{Commit-Validity}, and the CL timeout-triggered abort to \emph{Abort-Validity}.

\subsection{Atomicity Definition}
\label{sec:atomicity-definition}

Before describing the protocol components, we formally define what it means for a cross-chain transaction to be atomic. This definition will be used throughout our analysis in \secref{sec:analysis-protocol-correctness} to establish correctness properties.

\begin{definition}[Cross-Chain Atomicity]
\label{def:cross-chain-atomicity}
A cross-chain transaction $T = (t_1, t_2, \ldots, t_k)$ is atomic if and only if:
\[
\forall i,j \in [1,k]: \quad \status(t_i) = \status(t_j)
\]
where $\status(t_i) \in \{\textsf{success}, \textsf{failure}\}$ denotes the execution status of transaction $t_i$ on chain $C_i$.
That is, either all transactions in $T$ succeed, or all fail. No partial execution is allowed.
\end{definition}

This definition ensures that the fundamental property of atomicity is maintained across all chains involved in a CAT. The protocol must guarantee that this property holds regardless of network conditions, Byzantine behavior, or component failures.

\subsection{Components}
\label{sec:components}

We first identify the components that are required for the \cat{} protocol. We show an overview of the components that are involved in \figref{fig:high-level-complete-system}. We note that this is a minimal set of components, some of which may be merged, such as the executor and the resolver.

\vspace{0.3cm}

\begin{figure}[htbp]
  \centering
  \includegraphics[width=0.49\textwidth]{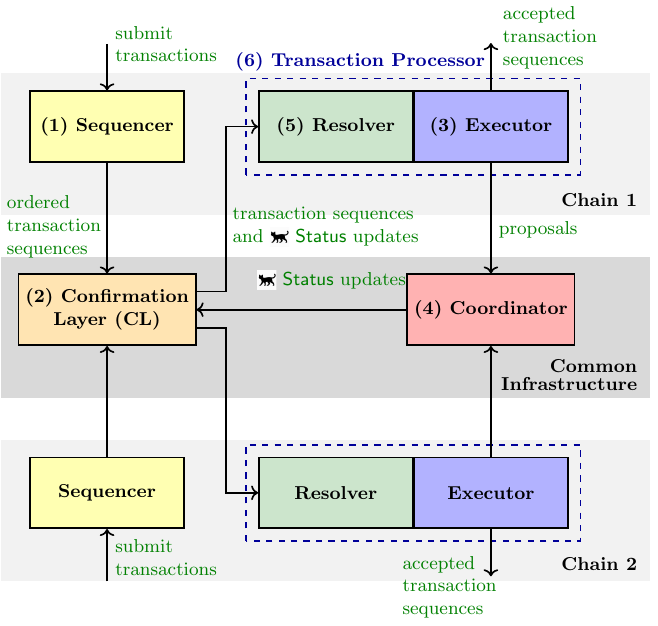}
  \caption{High-level overview of the involved system components in a \cat{} protocol with two chains.}
  \label{fig:high-level-complete-system}
\end{figure}

\subsubsection{Sequencer}
In order to submit a \cat{}, the user passes a transaction to the sequencer of a chain. We assume this sequencer has the right to propose a sequence of transactions to be executed, including \cat{}s. Note that this implies the sequencer of a given chain may also propose transactions for other chains that it interoperates with, as a \cat{} contains transactions for multiple chains.\footnote{We do not address how such a sequencer scheme is implemented; however, a simple approach would be to operate a \emph{shared sequencer}.

A shared sequencer also provides \emph{censorship resistance}: if an individual chain's sequencer censors a transaction, the CL can still include it. Safety derives from the CL and TP proofing, not the sequencer.}

\subsubsection{Confirmation layer (CL)}
Transactions that are provided by the sequencers are recorded on a confirmation layer. 
Moreover, since \cat{}s live across multiple chains, we require that the confirmation layer spans all involved chains, and we require consensus on inclusion of transactions from all chains.\footnote{The confirmation layer may also be considered as a commitment layer, as the sequencer and, as we will see later, the coordinator, commit on this layer.} 
We also record the status of the \cat{} on the confirmation layer.

We assume the CL operates as a BFT replicated state machine~\cite{lamport1998parttime,castro1999practical,schneider1990smr} under partial synchrony~\cite{dwork1988consensus}, with three minimal assumptions:
\begin{enumerate}
  \item \textbf{Safety:} BFT agreement among $2f{+}1$ out of $3f{+}1$ nodes ensures no two honest nodes disagree on CL state.
  \item \textbf{Liveness under partial synchrony:} After GST, honest proposals are included within a bounded number of rounds.
  \item \textbf{Censorship resistance:} Transactions observed by $f{+}1$ honest nodes are eventually included.
\end{enumerate}
Incentive alignment follows standard staking/slashing models as deployed in production BFT chains.

\subsubsection{Executor}
Each chain operates an executor that is responsible for executing transactions and updating the state. Since learning about the (potential) success or failure of a transaction requires execution (but not necessarily applying the changes to the state), the executor is essential to determining the status of \cat{}s and relaying that information.

We assume that the executor may deliver messages asynchronously. While asynchronous delivery could compromise liveness in a naive implementation of our protocol, it does not compromise safety. Moreover, through the introduction of timeouts, we can ensure that the protocol is live (dependent on the confirmation layer) even in the presence of asynchronous delivery.

\begin{figure*}[ht]
  \centering
  \includegraphics[width=0.75\textwidth]
  {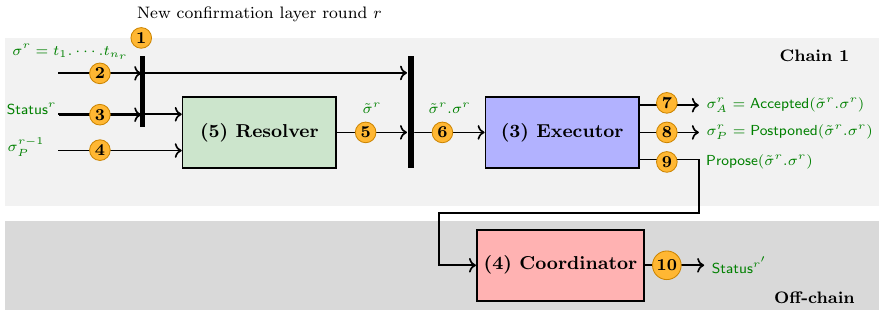}
  \caption{High-level overview of 
  transaction processing on a single 
  chain. The 10 steps shown are 
  described in detail in \secref
  {sec:naive-cat-protocol}.}
  \label{fig:high-level-one-chain}
\end{figure*}

\subsubsection{Coordinator}
An impossibility result by \cite{Zamyatin2021} shows that it is impossible to agree on the status of a \cat{} without a third party. We call this third party the \emph{coordinator}. The coordinator is responsible for facilitating the resolution of \cat{}s across chains. It takes as input proposals from the executors and returns a solution to the involved chains. 

One important aspect becomes apparent when considering the role of the coordinator. The sequencer-proposed transaction sequence is not final—the actual sequence of transactions that gets accepted and applied to the state depends on additional information, specifically the status response from the coordinator. To reconstruct the current state of a chain, we need both the sequence of ordered transactions and a record of all coordinator decisions made for CATs, since these decisions determine which transactions were actually successful or skipped.
As we will later see, the coordinator must provide the solution to the same confirmation layer as the sequencers, as we require strict consensus on the status delivery for all involved chains.

We assume that the coordinator may deliver messages to the confirmation layer asynchronously. While asynchronous delivery could compromise liveness in a naive implementation of our protocol, it does not compromise safety. Moreover, through the introduction of timeouts, we can ensure that the protocol is live (dependent on the confirmation layer) even in the presence of asynchronous delivery.

\begin{mdframed}[backgroundcolor=gray!10, linewidth=0.5pt]
  The coordinator is a third party that commits to the success (or failure) of executions of transactions. This is notably different from the sequencer, which commits ``only'' to the inclusion of transactions.
\end{mdframed}

\subsubsection{Resolver}

Two types of information are passed to the chains at this point. First, the sequencer passes the sequence of transactions to be executed. However, these are not necessarily eligible for execution due to (dependencies on) \cat{}s. Second, the coordinator passes the resolution of the \cat{} to the chains through the CL. Together, these two input streams provide an ultimate sequence of transactions which are accepted for execution. Since we would like the executor to be specialized on execution, we distinguish between the resolver and the executor, where the resolver is responsible for resolving the input stream provided by the coordinator and the sequencer.

\subsubsection{Transaction Processor (TP)}

We define the transaction processor as the component that operates both a resolver and an executor.

\begin{figure*}[h]
  \centering
  \includegraphics[width=0.75\textwidth]{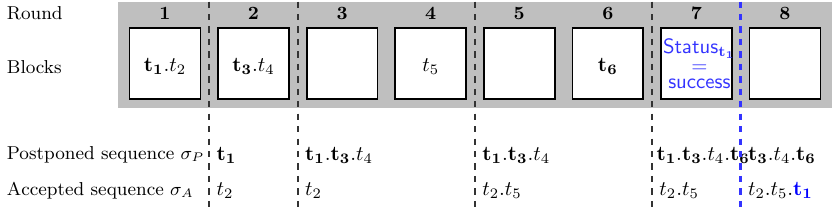}
  \caption{Example of chain activity with the naive \cat protocol. Transactions ${t_1}$, ${t_3}$, and ${t_6}$ are part of \cat{}s, whereas transactions $t_2, t_4$ and $t_5$ are normal transactions. ${t_3} \directDependsOn{s} {t_1}$, ${t_4} \directDependsOn{s} t_3$, and $t_4 \directDependsOn{s} t_2$. The \textcolor{blue}{\textbf{status}} of ${t_1}$, reported by the coordinator after 6 blocks, is $\textcolor{blue}{\success}$. The relationships between the transactions are further illustrated in \figref{fig:dependencies}. }
  \label{fig:example-one-chain}
\end{figure*}

\subsection{Naive \cat{} protocol}
\label{sec:naive-cat-protocol}

We now describe the algorithms in detail for the components Resolver, Executor, and Coordinator. The Confirmation Layer orders the incoming messages into blocks, which is a component external to the core of the \cat{} protocol, and thus will not be discussed here.

\figref{fig:high-level-one-chain} shows the transaction processing for one of the involved chains. The itemized list below corresponds with the numbered steps in the figure. 
We also provide the corresponding algorithms for the Resolver (Algorithm~\ref{alg:resolver}), the Executor (Algorithms~\ref{alg:executor-with-timeouts} and \ref{alg:executor-depth-limited}), and the Coordinator (Algorithm~\ref{alg:coordinator}). The steps are as follows:

\begin{enumerate}
  \item \textbf{Round initialization.} The consensus layer starts a new round $r$. It may have received any or none of the inputs mentioned in points 2. or 3. If there is no input, the round is skipped (or an empty block is produced).
  \item \textbf{Transaction input.} A set of transactions $S^r=\{t_1,..,t_{n_r}\}$ with sequence $\sigma^{r} =t_1. t_2. \cdots . t_{n_{r}}$ is provided by a sequencer.
  \item \textbf{Status updates.} A status update $\status^{r-1}$ is provided by the \emph{coordinator} and may indicate that the status of some \cat{}s in the set $S_P^{r-1}$ has been decided (\failure or \success).
  \item \textbf{Pending transactions.} The transactions that were left \emph{unprocessed/unfinalized} at previous rounds (they are waiting for an \cat{} to be finalized) are in the set $S_P^{r-1}$ with the sequence $\sigma_P^{r-1}$. 
  \item \textbf{Status resolution.} We (informally) define transaction sequence $\tilde{\sigma}^{r}$ to be the result of $\sigma_P^{r-1}$ where the status of the transactions in $S_P^{r-1}$ is resolved: for instance if we had a transaction  in $S_P^{r-1}$ that was either \success/\failure, we resolve its status to the one (\success/\failure) given by $\status_{r-1}$.
  \item \textbf{Sequence construction.} The sequence $\tilde{\sigma}^{r} . \sigma^{r}$ is built to be processed at round $r$.
  \item \textbf{Transaction reordering.} The executor splits this sequence of transactions into two transaction sequences: 
  
  \vspace{0.2cm}
  \hspace{0.5cm}The transaction set $S_A^r$ with sequence $\sigma_A^{r} = \accepted(\tilde{\sigma}^{r} . \sigma^{r})$, are transactions that can be finalized immediately in round $r$.

  \hspace{0.5cm}The transaction set $S_P^r$ with sequence $\sigma_P^{r} = \postponed(\tilde{\sigma}^{r} . \sigma^{r})$, which are transactions that cannot be finalized in round $r$ (they are \cat{}s with a still unresolved status, or they depend on these). 
  \vspace{0.2cm}

 This allows for new order of transactions as the one provided by the sequencer. However, it does not affect the outcome of the original sequence, see Lemma~\ref{thm:parallel-computation-independent-transactions}. \footnote{In \secref{sec:improvements} we discuss improvements that permit the reordering by the coordinator in a way that could lead to different outcomes -- a trade-off that gives the coordinator some control over the order of transactions in exchange for liveness guarantees.}
  \item \textbf{Postponed transactions.} The sequence $\sigma_P^{r}$ is propagated for the next round. In practice the Resolver and the Executor are the same entity (i.e., the Transaction Processor) and thus the sequence is just updated locally.
  \item \textbf{Proposal provision.} The executor computes a proposal $\propose(\tilde{\sigma}^{r} . \sigma^{r})$ for how to resolve \cat{}s in $\tilde{\sigma}^{r} . \sigma^{r}$. The proposal may take arbitrary forms, for example we may directly output \success/\failure for a \cat{} if its status can be computed, or also indirect dependencies including their possible outcomes could be computed and represented through a dependency graph, see \secref{sec:dependencies}. The $\propose(\cdot)$ set is communicated to the coordinator, who is in charge of resolving the \emph{common outcome} of \cat{}s and can be combined with the other $\propose(\cdot)$ sets of the other chains to determine the status that will be communicated to both chains via $\status^{r'}, r' \geq r $.
  \item \textbf{Coordinator resolution.} The coordinator combines the proposals from all chains to determine the statuses of the \cat{}s. This results in a status update $\status^{r'}$, $r' \geq r$, that is processed by the resolver at the next round ($r+1$) or some later round (if delivery to the confirmation layer is delayed).
\end{enumerate}

We illustrate the effects of the naive \cat{} protocol for the set of transactions from Table~\ref{tab:transaction-set} in \figref{fig:example-one-chain}. $t_1$ is a \cat{} and thus gets postponed. $t_2$ and $t_5$ are independent regular transactions and are accepted and executed creating the accepted sequence $t_2.t_5$. In the scenario shown the delivery of the status update for the \cat{} $t_1$ takes 6 blocks. $t_1$ is then successfully executed and the accepted sequence is updated to $t_2.t_5.t_1$.


\subsection{Complete CAT protocol}
\label{sec:improvements}

As we have seen in \secref{sec:analysis-protocol-correctness}, the \cat{} protocol can invoke a high level of blocking, if approached naively. Moreover, this blocking can lead to liveness issues, which is not acceptable. In this section we discuss possible improvements, at the cost of empowering the coordinator. This empowerment of the coordinator is to a degree, where it gains the ability to reorder some of the transactions in a way that could affect the outcome. Consequently, we must consider to incentivize and ensure honest behavior from the coordinator, see also \secref{sec:analysis-protocol-correctness}.

\begin{figure*}[h]
  \centering
  \includegraphics[width=0.75\textwidth]{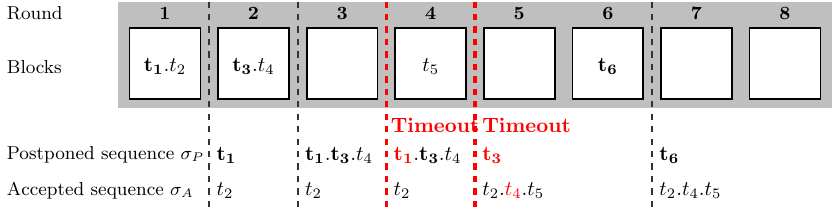}
  \caption{Improvement through \textcolor{red}{timeouts}. Example as in \figref{fig:example-one-chain}, but with a timeout of $\Delta=2 \,\,rounds$. Since the coordinator did not respond two rounds after the blocking transaction $\boldsymbol{t_1}$ was added to the chain, the transaction is \textbf{skipped}. Similarly, $\boldsymbol{t_3}$ also times out shortly after. $t_4$ is a regular transaction and not blocking, therefore, it is accepted and the accepted sequence is updated to $t_2.t_4$. }
  \label{fig:example-one-chain_v2}
\end{figure*}

\subsubsection{Timeouts.}
\label{sec:cat-protocol:timeouts}

As can be seen in \figref{fig:dependencies}, transactions may have deep indirect dependencies. If a transaction is blocking for a substantial amount of time, such as is the case in \figref{fig:example-one-chain}, it may lead to a large number of transactions being postponed, which can lead to many transactions having to be kept in memory and long delays. Ultimately this may even result in liveness issues and an attack vector. For this reason, we impose a timeout $\Delta$ for when the coordinator must reply. The detailed algorithm is provided in Algorithm~\ref{alg:executor-with-timeouts}.

We must ensure that chains do not determine the outcome of a \cat{} based on these timeouts locally, as this would invalidate coordination between chains, and thus violate atomicity. 
This is similar to the approach described in Avalon \cite{Cai2024cats-through-state-layers}, where chains communicate directly with each other.
Similarly to our approach, chains can abort through timeouts, which could lead to inconsistencies when relying on local timestamps. To overcome this challenge, it is proposed to rely on an \emph{upper layer blockchain}, which provides synchronized authenticated timestamps across all underlying chains. A timeout is triggered when the number of blocks from the upper chain exceeds a predetermined value before all relevant messages are published to the upper chain.
Consequently, we propose that timeouts must be handled via the common CL.

We note that this operation is safe, see \secref{sec:analysis-protocol-correctness}, since the CL, which operates across all chains, ensures that the transactions that are part of a \cat{} are \textbf{skipped} at all involved chains at the correct round.

This approach permits us to overcome concerns with respect to the liveness of the coordinator. However, it does also add reliance on the coordinator, who may be Byzantine and could in principle utilize the timeout to reorder the transaction stream. Hence, particular care should be given to the configuration of the timeout $\Delta$.

We illustrate the effects of the timeout mechanism for the set of transactions from Table~\ref{tab:transaction-set} in \figref{fig:example-one-chain_v2}. With a timeout of $\Delta=2$ rounds, the \cat{}s $t_1$ and $t_3$ time out and get skipped after 2 blocks, allowing $t_4$ to be accepted and executed creating the accepted sequence $t_2.t_4.t_5$. This demonstrates how timeouts prevent indefinite blocking.

\subsubsection{Limiting the dependency depth.}
\label{sec:cat-protocol:limiting-dependency-depth}

As can be seen in \figref{fig:dependencies}, transactions may have several dependencies, and thus the dependency depth (which considers only the pending set $\sigma_P$) may grow substantially if not resolved fast enough.
Moreover, each \cat{} introduces a set of new outcomes, dependent on whether the outcome is \success or \failure. This can result in an exponential number of superposition states with the number of \cat{}s, in particular if the \cat{}s or transactions are highly dependent on each other. 
Not only could this result in long delays, but it could also enable resource-exhaustion attack vectors. 

As we show in the previous example, we may reorder the sequencer transaction stream in order to improve the protocol properties, such as liveness. Similarly, here we may reorder the transaction stream proposed by the sequencer in order to limit the dependency depth to prevent resource-exhaustion attacks. However, and analogously to the introduction of timeouts, the coordinator is given additional power, as it can invoke delays that result in reordering of the transaction stream.

\medskip

In \secref{sec:dependencies} we introduced the dependency DAG $\calD$ for a given state $s$ and a sequence of transactions $\sigma$. We show an example of the dependency DAG including the CAT dependency depth in \figref{fig:dependencies_with_depth}. We only consider the pending set $\sigma_P$ for the DAG, \i.e. $\calD=\calD(s,\sigma_P)$. Transactions that can be readily accepted are added to the accepted sequence $\sigma_A$, which we can ignore for the dependency DAG.
For the algorithm, we add a new parameter $\texttt{maxDepth}$ and a new sequence of \emph{ignored} transactions $\calI$.
The approach that we propose limits the CAT dependency depth and protects against resource-exhaustion attacks.
It also reduces the number of transactions that need to be considered for postponement, which reduces the memory and latency of the protocol.

We augment the executor with dependency depth limiting as shown in Algorithm~\ref{alg:executor-depth-limited}. The changes are highlighted in blue. The following rules are applied for a given transaction $t$:
\begin{enumerate}
  \item If $t$ is a \cat{} and $\depth_{s^{r-1},\sigma_P^r}(t) > \texttt{maxDepth}$ we add $t$ to the ignored set $\mathcal{I}$.
  \item If $t$ is not a \cat{} and $\depth_{s^{r-1},\sigma_P^r}(t) > \texttt{maxDepth}$ we add it to the ignored set $\mathcal{I}$.
  \item Else if $t$ is not a \cat{} and $\depth_{s^{r-1},\sigma_P^r}(t) > 0$ we add it to the pending set $\sigma_P$.
  \item Else we process $t$ as normal.
\end{enumerate}

\begin{figure}[htbp]
  \centering
  \includegraphics[width=0.49\textwidth]{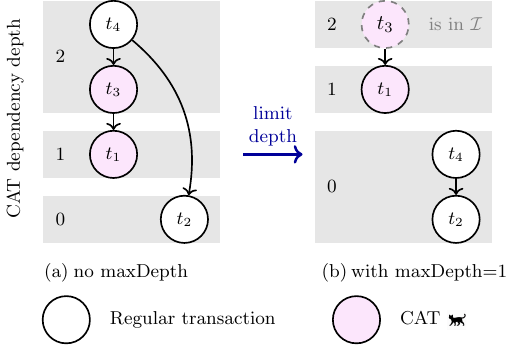}
  \caption{Effect of a depth limit $\texttt{maxDepth}=1$ on the dependency graph from \figref{fig:dependencies}. Since $\depth_{(s,t_1.t_2)}(t_3)=1$, $t_3$ is added to the ignored sequence $\mathcal{I}$ (gray dashed border). Thus, $\depth_{(s,t_1.t_2)}(t_4)=0$.}
  \label{fig:dependencies_with_depth}
\end{figure}

\begin{figure*}[h]
  \centering
  \includegraphics[width=0.75\textwidth]{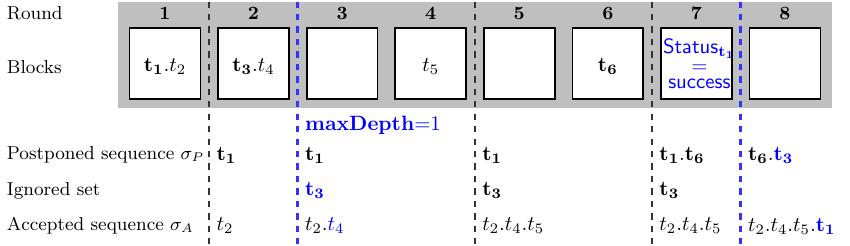}
  \caption{Improvement through \textcolor{blue}{limiting the CAT dependency depth} with $\textbf{maxDepth}=1$. 
  We use the transactions from Table~\ref{tab:transaction-set}, with dependency graph illustrated in \figref{fig:dependencies_with_depth}.
  Since $\depth_{(s,t_1.t_2)}(t_3)=2$, we add $\boldsymbol{t_3}$ to the ignored sequence $\calI$. 
  Since now $\depth_{(s,t_1.t_2)}(t_4)=0$, it is directly added to the accepted sequence $\sigma_A$.
  Eventually $\boldsymbol{t_1}$ is accepted and $\boldsymbol{t_3}$ gets computed and added to the postponed sequence $\sigma_P$.}
  \label{fig:example-one-chain_v3}
\end{figure*}

\noindent We show an example of the application of the algorithm in \figref{fig:example-one-chain_v3}. With a depth limit of $\texttt{maxDepth}=1$, the \cat{} $t_3$ is added to the ignored sequence $\calI$ and $t_4$ is directly added to the accepted sequence $\sigma_A$. Eventually $t_1$ is accepted and $t_3$ gets computed and added to the postponed sequence $\sigma_P$.


\section{Analysis of Protocol Correctness}
\label{sec:analysis-protocol-correctness}

The analysis of our CAT protocol focuses on three critical aspects: first, establishing the foundation that our protocol can be made minimally-blocking through safe reordering of independent transactions; second, formal safety and liveness guarantees with respect to trusted entities that verify protocol correctness; and third, detailed analysis of safety and liveness across the protocol components described in Section~\ref{sec:components}. We establish that the protocol provides strong safety guarantees through cryptographic protections and consensus mechanisms, while ensuring liveness through timeout mechanisms and independent execution of non-dependent transactions.

\subsection{Minimally-blocking protocol.}
\label{sec:minimally-blocking-cat}

Before defining safety and liveness properties, we establish a crucial foundation: our protocol achieves minimal blocking by safely reordering independent transactions. This means that transactions without dependencies can execute immediately with zero blocking time, while dependent transactions are guaranteed to complete within a bounded timeout. This property is essential for protocol efficiency and directly supports the safety analysis that follows, as it shows that regular transaction ordering can be optimized without compromising correctness.

\begin{proposition}[Minimally-blocking Protocol]
\label{prop:minimally-blocking-protocol}
Given a timeout mechanism, our protocol that permits reordering of independent transactions achieves minimal blocking with two aspects:

\medskip

\textbf{Committed regular transactions}: The blocking time for any transaction committed to the CL is bounded by $\Delta$:
\[
\forall t \in \mathcal{T}: \quad \text{BlockingTime}(t) \leq \Delta
\]

\textbf{Independent transactions}: No blocking time at all:
\[
\forall t \in \mathcal{T}: \quad \text{Independent}(t) \implies \text{BlockingTime}(t) = 0
\]
where $\text{BlockingTime}(t)$ is the time a transaction waits before being executed (\failure or \success), and $\text{Independent}(t)$ is a predicate that is true if $t$ has no dependencies.
\end{proposition}

\begin{proof}
By Lemma~\ref{lemma:out-of-order-execution-independent-transactions} from Section~\ref{sec:model}, independent transactions can be executed in parallel without affecting correctness and the final state.
For any transaction $t$, the protocol achieves minimal blocking through:
\begin{enumerate}
    \item Zero blocking for independent transactions: If $t$ has no dependencies, it executes immediately with $\text{BlockingTime}(t) = 0$.
    \item Bounded blocking for CATs: If $t$ is part of a \cat{} $T$, the timeout mechanism ensures $\text{BlockingTime}(t) \leq \Delta$.
    \item Bounded blocking for dependent transactions: If $t$ has dependencies, the timeout mechanism ensures $\text{BlockingTime}(t) \leq \Delta$, since all CATs that are ordered before $t$ are blocked at most until the timeout $\Delta$.
\end{enumerate}

\noindent This ensures that independent transactions have no blocking time while dependent transactions are bounded by $\Delta$, achieving true minimal blocking.
\end{proof}

\subsection{Formal Model and Definitions}
\label{sec:formal-model}

We build upon the formal model established in Section~\ref{sec:model}. 
Before defining the key properties, we introduce the \emph{non-proposing Transaction Processor} (\nonTP), which are client nodes with which we measure our safety and liveness properties. The \nonTP is in essence a TP but does not interact with the coordinator. It can be operated as trusted since any user may deploy it, and it only relies on read access to the CL. Moreover, since the \nonTP constitutes a subset of the tasks of the TP, all described guarantees derived herein hold true for the TP. We will distinguish between \nonTP{}s on the same chain and \nonTP{}s on different chains. We, therefore, define $\mathcal{N}_{C_i}$ as the set of \nonTP{}s on chain $C_i$.

We use the timeout parameter $\Delta$ (already defined in Section~\ref{sec:cat-protocol:timeouts}) which is measured in rounds and used in the algorithms (see Appendix~\ref{sec:algo:executor-with-timeouts}).
State transitions follow the $\Next$ function defined in Section~\ref{sec:model}.

We now define key properties that our protocol must satisfy:

\begin{definition}[Safety]
\label{def:safety}
A protocol is safe if for any CAT $T$, all \nonTP{}s maintain consistent views of $T$'s status, both within the same chain and across different chains. Specifically:

\medskip
\noindent\textbf{Intra-chain consistency}: For any chain $C_i$ and any two \nonTP{}s $N_1, N_2$ on $C_i$:

\vspace{0.2cm}
\noindent$\forall T, \forall C_i, \forall N_1, N_2 \in \mathcal{N}_{C_i}$: 
\[
\begin{gathered}
\quad \status_{N_1}(T) = \status_{N_2}(T)
\end{gathered}
\]

\noindent\textbf{Inter-chain consistency}: For any two \nonTP{}s $N_1$ on chain $C_i$ and $N_2$ on chain $C_j$ (where $i \neq j$):

\vspace{0.2cm}
\noindent$\forall T, \forall C_i, C_j, \forall N_1 \in \mathcal{N}_{C_i}, \forall N_2 \in \mathcal{N}_{C_j}$:
\[
\begin{gathered}
\quad \status_{N_1}(T) = \status_{N_2}(T)
\end{gathered}
\]

\end{definition}

\begin{definition}[Liveness]
\label{def:liveness}
A protocol is live if for any CAT $T$, there exists a finite timeout threshold $\Delta$ such that:
\[
\forall T: \quad \exists \Delta: \quad \text{resolution}(T) \leq \Delta
\]
where $\text{resolution}(T)$ is the time until $T$ reaches a final state, measured in rounds. This liveness guarantee ensures that all \nonTP{}s will eventually observe a final status for $T$ within $\Delta$ rounds.
\end{definition}

\subsection{Confirmation Layer}
\label{sec:confirmation-layer}

The CL is assumed to provide a decentralized Byzantine fault-tolerant consensus mechanism that secures both the order and content of all transactions (CATs and regular transactions). It determines time progress through rounds, and timeouts are measured in CL rounds. The CL serves as the authoritative source for all transaction-related decisions. All TPs and \nonTP{}s must read from the CL to determine the correct status of transactions and proceed with state transitions.

\begin{proposition}[CAT Status Focus]
\label{prop:cat-status-focus}
Since the CL provides safe consensus on transaction content and order, it is sufficient to focus on whether \nonTP{}s maintain consistency on CAT statuses. Regular transaction ordering does not affect the safety properties that we require to prove for \nonTP{}s.
\end{proposition}

\begin{proof}
By Proposition~\ref{prop:minimally-blocking-protocol}, independent regular transactions can be executed out of order without affecting final state consistency.
Since the CL is safe and secures transaction order and content, for \nonTP{}s this means:
\begin{enumerate}
    \item Regular transaction ordering is handled safely by the \nonTP{}'s executor's reordering logic.
    \item The safety of the state updates for all types of transactions follows from the \nonTP{}'s executor's correctness.
    \item We still need to prove that \nonTP{}s maintain consistency on CAT statuses.
\end{enumerate}

\noindent Therefore, \nonTP{} safety properties can focus exclusively on CAT status consistency.
\end{proof}

\begin{proposition}[Safety with respect to the CL]
\label{prop:cl-safety}
The CL satisfies the safety property defined in Definition~\ref{def:safety}. Specifically, for any CAT $T$ and any two \nonTP{}s $N_1, N_2$ on any chain $C_i$, if both read the CL at round $r$, then:
\[
\status_{N_1}(T, r) = \status_{N_2}(T, r)
\]
where $\status_N(T, r)$ denotes the status of CAT $T$ as observed by \nonTP $N$ at round $r$.
\end{proposition}

\begin{proof}
By assumption, the CL implements a Byzantine fault-tolerant consensus protocol with safety property $S$. 
For any CAT $T$, the CL records status updates through consensus. By the safety property $S$ of the underlying consensus protocol:
\[
\forall N_1, N_2 \in \mathcal{N}_{C_i}, \forall r: \quad \text{CL}_r^{N_1} = \text{CL}_r^{N_2}
\]
where $\text{CL}_r^N$ is the CL state at round $r$ as observed by \nonTP $N$.
Since $\status_N(T, r)$ is derived from $\text{CL}_r^N$, we have:
\[
\status_{N_1}(T, r) = \status_{N_2}(T, r)
\]

\noindent This ensures that all \nonTP{}s observe consistent CAT statuses, maintaining safety across the system.
\end{proof}

\begin{proposition}[Liveness with respect to the CL]
\label{prop:cl-liveness}
Let $L_{\text{CL}}$ be the liveness property of the CL consensus protocol. If $L_{\text{CL}}$ holds, then for any CAT $T$ and any \nonTP $N$ on any chain $C_i$, there exists a finite round $r_T$ such that:
\[
\forall r \geq r_T: \quad \status_N(T, r) \neq \bot
\]
where $\bot$ denotes an undefined status.
\end{proposition}

\begin{proof}
By assumption, the CL implements a Byzantine fault-tolerant consensus protocol with liveness property $L_{\text{CL}}$. This means that for any valid input, the CL eventually produces an output within a finite number of rounds.
For any CAT $T$, the CL must eventually record a status update. The absence of a \status update from the coordinator within timeout is equivalent to a \status update with \failure. By the liveness property $L_{\text{CL}}$:

\vspace{0.2cm}
\noindent$\exists r_T: \quad \forall r \geq r_T, \forall N \in \mathcal{N}_{C_i}$: 
\[
\text{CL}_r^N \text{ contains the \status of } T
\]
Since $\status_N(T, r)$ is derived from $\text{CL}_r^N$, we have:
\[
\forall r \geq r_T: \quad \status_N(T, r) \neq \bot
\]
This ensures that all \nonTP{}s eventually observe defined CAT statuses, maintaining liveness across the system.
\end{proof}

The CL's role as the authoritative source of truth enables a key system-wide property: atomicity across chains. Since all chains read from the CL to determine CAT statuses, the CL ensures that cross-chain transactions maintain atomicity.

\begin{proposition}[Atomicity Guarantee]
\label{prop:atomicity}
For any CAT $T$, the protocol ensures that $T$ satisfies the atomicity property defined in Definition~\ref{def:cross-chain-atomicity}.
\end{proposition}

\begin{proof}
By Proposition~\ref{prop:cl-safety}, all \nonTP{}s eventually agree on the status $\alpha$ of $T=(t_1, \ldots, t_k)$ through the CL.
For any chain $C_i$ involved in $T$, the state transition follows:
\[
s_i \xrightarrow{\ t_i/\alpha \ } s_i'
\]
where $s_i' = s_i$ if $\alpha = \textsf{failure}$ (no state change).
By the CL's global consistency (Proposition~\ref{prop:cl-safety}), all \nonTP{}s observe the same status $\alpha$ for $T$, and thus its constituent transactions $t_i$ have the same status $\alpha$. Therefore:
\[
\forall i,j \in [1,k]: \quad \status(t_i) = \alpha = \status(t_j)
\]
This ensures that $T$ satisfies the atomicity property from Definition~\ref{def:cross-chain-atomicity}, maintaining consistency across all chains.
\end{proof}


\subsection{Coordinator}
\label{sec:coordinator}

A third-party coordinator is required to agree on the status of \cat{}s \cite{Zamyatin2021}.
However, the coordinator may be faulty or Byzantine. To ensure safety, the coordinator must record its decisions on the CL for accountability as well as to ensure consistent global view. In addition, the decisions of the TPs must be protected by signature and these signatures must be contained in the status updates from the coordinator recorded on the CL. To ensure liveness and because the coordinator may become faulty, we require timeouts.

\begin{proposition}[Safety with respect to the Coordinator]
\label{prop:coordinator-safety}
Let $\mathcal{C}_{\text{coord}}$ be the coordinator and $\mathcal{T}_{\text{coord}}$ be the set of TPs involved in a CAT $T$. For any status proposal $\status$ from TP $i \in \mathcal{T}_{\text{coord}}$, the coordinator can only:
\begin{enumerate}
    \item Record $\status$ on the CL: $\text{CL}_h \leftarrow \status$
    \item Ignore $\status$: $\text{CL}_h \not\leftarrow \status$
\end{enumerate}
The coordinator cannot produce $\status' \neq \status$ such that $\text{CL}_h \leftarrow \status'$.
\end{proposition}

\begin{proof}
By Proposition~\ref{prop:cl-safety}, messages recorded on the CL are safe and cannot be altered by Byzantine actors.
For any status proposal $\status$ from TP $i$, the coordinator receives $\status$ with cryptographic signatures $\text{sig}_i(\status)$. The coordinator can only:
\[
\text{CL}_h \leftarrow \status \quad \text{or} \quad \text{CL}_h \not\leftarrow \status
\]

\noindent If the coordinator attempts to record $\status' \neq \status$, the cryptographic signatures $\text{sig}_i(\status)$ would not validate for $\status'$, making the status update invalid. By the CL's safety properties, invalid updates are rejected.
Therefore, the coordinator's actions are limited to recording the original proposal or ignoring it, ensuring that only authentic status proposals are recorded.
\end{proof}

\begin{proposition}[Liveness with respect to the Coordinator]
\label{prop:coordinator-liveness}
For any CAT $T$, liveness is ensured through the CL's timeout mechanisms, not through independent coordinator timeouts.
\end{proposition}

\begin{proof}
By Proposition~\ref{prop:cl-liveness}, the CL ensures that all CATs eventually reach a defined status within a finite number of blocks.
The coordinator's role is to facilitate agreement on CAT statuses, but the actual timeout enforcement and global consensus comes from the CL. If the coordinator fails to provide a status update, the timeout mechanism implemented in the \nonTP{} is enabled with the help of the CL consensus.
Therefore, liveness is maintained through the CL's guarantees, with the coordinator acting as a facilitator rather than the source of timeout enforcement.
\end{proof}


\subsection{Transaction Processor}
\label{sec:transaction-processor}

If the TP proposes a status to the coordinator for a given \cat{} $T$, this status proposal must be correct.\footnote{For example, if the TP proposes \success, although the execution by a \nonTP would fail, then a user could be exploited by giving up assets on another chain.} 
However, a centralized TP may become Byzantine. Thus, the correctness of the status proposals must be protected by providing a degree of finality, e.g. cryptographically (ZK) or crypto-economically (through attestations of a staked committee).

\begin{proposition}[Safety with respect to the TP]
\label{prop:safety-progress}
Let $\mathcal{P}_{\text{TP}}$ be the protection mechanism (cryptographic or crypto-economic) for TP proposals. For any CAT $T$ and TP $i$, if $\mathcal{P}_{\text{TP}}$ is valid, then:

\vspace{0.2cm}
\noindent $ \forall \status \in \{\textsf{success}, \textsf{failure}\}$:
\[
\text{Verify}_{\mathcal{P}_{\text{TP}}}(\status, T, i) \in \{\text{true}, \text{false}\}
\]
where $\text{Verify}_{\mathcal{P}_{\text{TP}}}$ validates the protection mechanism for proposal $\status$.
\end{proposition}
\begin{proof}
By assumption, the protection mechanism $\mathcal{P}_{\text{TP}}$ provides either:
\begin{enumerate}
    \item \textbf{Cryptographic protection}: Zero-knowledge proofs or digital signatures that are computationally infeasible to forge.
    \item \textbf{Crypto-economic protection}: Staked attestations where dishonest behavior results in economic penalties.
\end{enumerate}

\noindent For any status proposal $\status$ from TP $i$, the verification function $\text{Verify}_{\mathcal{P}_{\text{TP}}}$ ensures:
\[
\begin{gathered}
\text{Verify}_{\mathcal{P}_{\text{TP}}}(\status, T, i) = \text{true} \\[2pt]
\iff \status \text{ is honestly generated by TP } i
\end{gathered}
\]

\noindent If a TP attempts to propose an incorrect status $\status' \neq \status$, the verification will fail:
\[
\text{Verify}_{\mathcal{P}_{\text{TP}}}(\status', T, i) = \text{false}
\]

\noindent This prevents violations of the \cat{} protocol by ensuring only honest status proposals are accepted.
\end{proof}

\noindent TPs may become faulty. However, we require that the \cat{} protocol guarantees eventual liveness, even if individual TPs fail. Since we cannot enforce liveness on individual chains, we require that the coordinator employs a timeout mechanism measured in CL block heights.

\begin{proposition}[Liveness with respect to the TP]
\label{prop:follower-executor-liveness}
For any CAT $T$ and faulty TP $i$, liveness is ensured through the CL's timeout mechanisms, not through independent TP mechanisms.
\end{proposition}

\begin{proof}
By Proposition~\ref{prop:cl-liveness}, the CL ensures that all CATs eventually reach a defined status within a finite number of blocks.
The TP's role is to execute transactions and provide status proposals, but the actual timeout enforcement and global consensus comes from the CL. If a TP $i$ fails to provide a status proposal for CAT $T$, the timeout mechanism implemented in the \nonTP{} is enabled with the help of the CL consensus.
Therefore, liveness is maintained through the CL's guarantees, with the TP acting as a facilitator for coordination rather than the source of timeout enforcement.
\end{proof}

\subsection{Summary}
\label{sec:summary}

\begin{theorem}[CAT Protocol Safety and Liveness]
\label{thm:cat-safety-liveness}
Under the assumptions of Byzantine fault tolerance for the CL and valid protection mechanisms for TPs (ZK or staking), our CAT protocol provides:
\begin{enumerate}
    \item \textbf{Safety}: All \nonTP{}s agree on CAT statuses and maintain atomicity
    \item \textbf{Liveness}: All CATs are eventually resolved within bounded timeout periods
    \item \textbf{Minimal Blocking}: Independent transactions execute in parallel, minimizing waiting time
\end{enumerate}
\end{theorem}

\begin{proof}
\textbf{Safety}: By Proposition~\ref{prop:cl-safety}, the CL provides consistent views to all \nonTP{}s given a secure consensus protocol. By Proposition~\ref{prop:atomicity}, this ensures atomicity of CAT \status updates across all chains.
By Proposition~\ref{prop:coordinator-safety}, the coordinator cannot produce an incorrect \status update, only ignore it.
By Proposition~\ref{prop:safety-progress}, the TP cannot produce an incorrect status given a valid protection mechanism.

\textbf{Liveness}: By Proposition~\ref{prop:minimally-blocking-protocol} regular independent transactions can be executed without blocking, while dependent regular transactions and CATs are blocked for at most $\Delta$ blocks given a valid timeout mechanism.
By Proposition~\ref{prop:cl-liveness}, the CL enables the required timeout mechanism and ensures that all CATs eventually reach a defined status within a finite number of blocks. By Proposition~\ref{prop:coordinator-liveness} and Proposition~\ref{prop:follower-executor-liveness}, the coordinator and the TP cannot impede this progress.

\textbf{Minimal Blocking}: By Proposition~\ref{prop:minimally-blocking-protocol} and Lemma~\ref{lemma:out-of-order-execution-independent-transactions} from Section~\ref{sec:model}, independent transactions execute immediately, while dependent transactions are blocked for at most $\Delta$ rounds.
\end{proof}

\section{Implementation and Evaluation}
\label{sec:evaluation}

We have implemented the protocol described in this paper as an open-source project, \textbf{Hyperplane}\footnote{Available at \url{https://github.com/movementlabsxyz/hyperplane}}. Hyperplane implements the protocol with timeouts and a $\texttt{maxDepth}=1$, including dependency management as formalized in our model.
The codebase is organized to closely follow the formal definitions and dependency structures introduced in Section~\ref{sec:model}. 

\subsection{Experimental Setup}

Our experimental evaluation was conducted on an AWS t3.xlarge instance with 4 vCPUs and 16 GB of memory.
We simulate two chains, each with 10000 accounts preloaded with sufficient tokens for experimentation. We randomly select accounts to send transactions, where each account uses the \xtxsend{sender}{receiver}{amount}
function described in Section~\ref{sec:motivation} to send tokens to another randomly selected account following a Zipf distribution $P(k) \propto 1/k^z$, where $k$ is the account rank and $z$ is the skewness parameter. 

\begin{table}[htbp]
\centering
\renewcommand{\arraystretch}{1.3}
\begin{tabular}{lc}
\arrayrulecolor{gray!30}
\hline
\rowcolor{gray!20}
\enspace \textbf{Parameter} \enspace & \enspace \textbf{Default Value} \enspace \\
\hline
Block interval (s) & 1 \\
\hline
Transactions per second (tps) & 100 \\
\hline
Zipf parameter ($z$) & 0.8 \\
\hline
Number of accounts & 10000 \\
\hline
CATs among total transactions & 50 tps (50\%) \\
\hline
Lifetime of CATs & 10 blocks \\
\hline
$\max$ CAT dependency depth & 1 \\
\hline
Delay of slowest chain ($D$) & 5 blocks \\
\hline
\end{tabular}
\vspace{0.3cm}
\caption{Default experiment parameters}
\label{tab:exp-params}
\end{table}

We set one chain to be fast with negligible delay and one chain to be slow. As the system performance is extensively affected by the slow chain, the content of the transactions in the fast chain is not relevant for the experiments, and we, thus, send transactions duplicated to both chains.
Our protocol implements each communication path separately: user $\rightarrow$ sequencer, sequencer $\rightarrow$ CL, CL $\rightarrow$ TP, TP $\rightarrow$ coordinator, coordinator $\rightarrow$ CL. For simulation purposes, we combine all delays into a single parameter between TP and coordinator, where the slowest chain dominates system behavior. Due to the non-synchronous nature of the system, additional delays on the order of milliseconds exist on communication paths between other components due to computational limitations. The delay parameter $D$ represents the combined message delays in this system.

For default parameters in Table~\ref{tab:exp-params}, we set $D$ to 50\% of the CAT lifetime, which seems a pessimistic scenario for a production system. The block interval is set to 1 second, which is a common block interval for recent blockchains.\footnote{In a production system, the confirmation layer produces blocks approximately every second, and execution is expected to be faster than a second \cite{tonkikh2025raptrprefixconsensusrobust}.}
The number of CATs per block is set to 50\% of the throughput, which simulates a system under high cross-chain load. 
We set $z=0.8$ to create a distribution skewed towards top accounts, which reflects common patterns in cryptocurrency systems; see e.g. \cite{kusmierz2021}. 

Each data point in the following figures is calculated by repeating the same simulation several times and calculating the average value at the end. We then display the mean, min, and max values.

\subsection{Results}
\label{sec:results}

In this section we present the results of our experiments evaluating the performance and scalability of our CAT protocol. Our results demonstrate that the protocol successfully resolves CATs in a timely manner and with a good success rate, with performance characteristics that adapt well to various network conditions and workload distributions.

\paragraph{Comparison with Avalon.}
Avalon~\cite{Cai2024cats-through-state-layers} is the only directly comparable cross-chain atomic transaction protocol. Both protocols require a constant number of coordination rounds. However, Avalon relies on per-pair IBC relayers, resulting in quadratic message complexity in the number of chains. The authors report that this causes latency to scale almost linearly in practice, as relayers slow down under the growing message load---15.1s at 3 chains, 34.7s at 6, and 59.2s at 9 chains without conflicts. Our protocol avoids this bottleneck: each chain's executor sends a proposal directly to the coordinator, which aggregates all proposals in one step and publishes a single decision to the CL, with linear message complexity. At $\sim$1s CL block intervals, this yields block-level finality independent of chain count. The two-chain experimental setup is sufficient because per-chain protocol behavior is identical regardless of participant count; adding more chains does not change the per-chain communication pattern.

\paragraph{Relation to standard benchmarks.}
Cross-shard and cross-chain coordination are structurally similar problems. TPC-C~\cite{tpcc2010} is the industry-standard benchmark for evaluating online transaction processing systems; its cross-warehouse transactions directly map to CATs: each touches state on two independent partitions that must commit atomically. The TPC-C standard specifies a natural cross-partition rate of approximately 11\% of total traffic (45\% NewOrder transactions with $\sim$10\% cross-warehouse, 43\% Payment transactions with 15\% cross-warehouse). Our evaluation already covers and exceeds this regime: \figref{fig:cat-ratio} varies $r_C$ from low values (where success rates exceed 90\%) through our stress-test default of 50\%---far beyond what standard workloads prescribe.


\subsubsection{Number of CATs per block.}
\label{sec:results:cat-ratio}

We modify $n_C$ (the number of CATs per block) by varying $r_C$ (the CAT ratio), which defines the number of CATs among the total number of transactions.
\figref{fig:cat-ratio} shows the effect on the success ratio of CATs. 

\begin{figure}[htbp]
    \centering
    \includegraphics[width=0.49\textwidth,trim={0 0 0 0cm},clip]{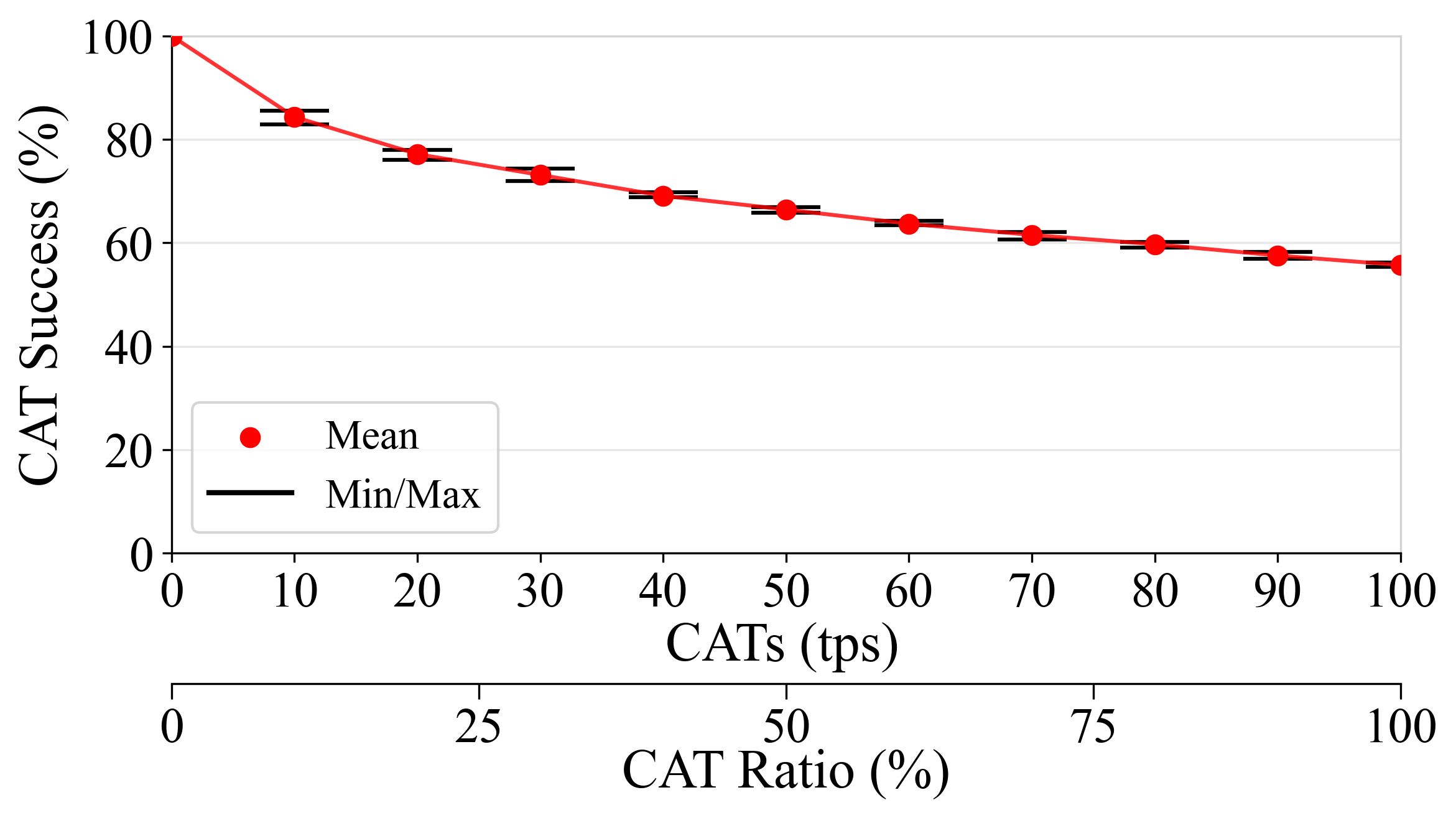}
    \caption{Average success of CATs with the number of CATs per block ($n_C$). We also display a secondary axis to put $n_C$ into context with the ratio of CATs to the total number of transactions ($r_C$).}
    \label{fig:cat-ratio}
\end{figure}

We display a secondary axis to indicate that we vary $n_C$ up to 100 tps, at which point only CATs are sent. The success ratio is significantly higher for low $n_C$ and decreases as $n_C$ increases. This is expected, as the protocol is designed to provide a non-interruptive experience for each chain's native transactions: regular transactions always make progress, at most after a short delay, regardless of pending CATs. This prioritization is enforced by the CAT timeout---a deliberately short window after which a CAT is aborted so that blocked native transactions can immediately proceed.

The declining success rate under high contention is not a cross-chain phenomenon---it is universal even on single blockchains. Large-scale empirical analysis of Solana shows overall transaction failure rates exceeding 75\% during peak congestion and contention~\cite{zheng2025solana}.

\figref{fig:cat-ratio-pending-combined} shows the impact of $n_C$ on different types of postponed transactions. As $n_C$ increases, we observe three distinct effects: (a) more regular transactions remain pending due to them waiting for a CAT to be resolved. Eventually, the number of postponed regular transactions decreases as most of the transactions are CATs. (b) More CATs remain in a resolving state as $n_C$ increases. (c) Increasingly more transactions are ignored due to dependency conflicts, highlighting the trade-off between CAT throughput and success rate.

\begin{figure}[htbp]
\centering
\includegraphics[width=0.49\textwidth,trim={0 0 0 0cm},clip]{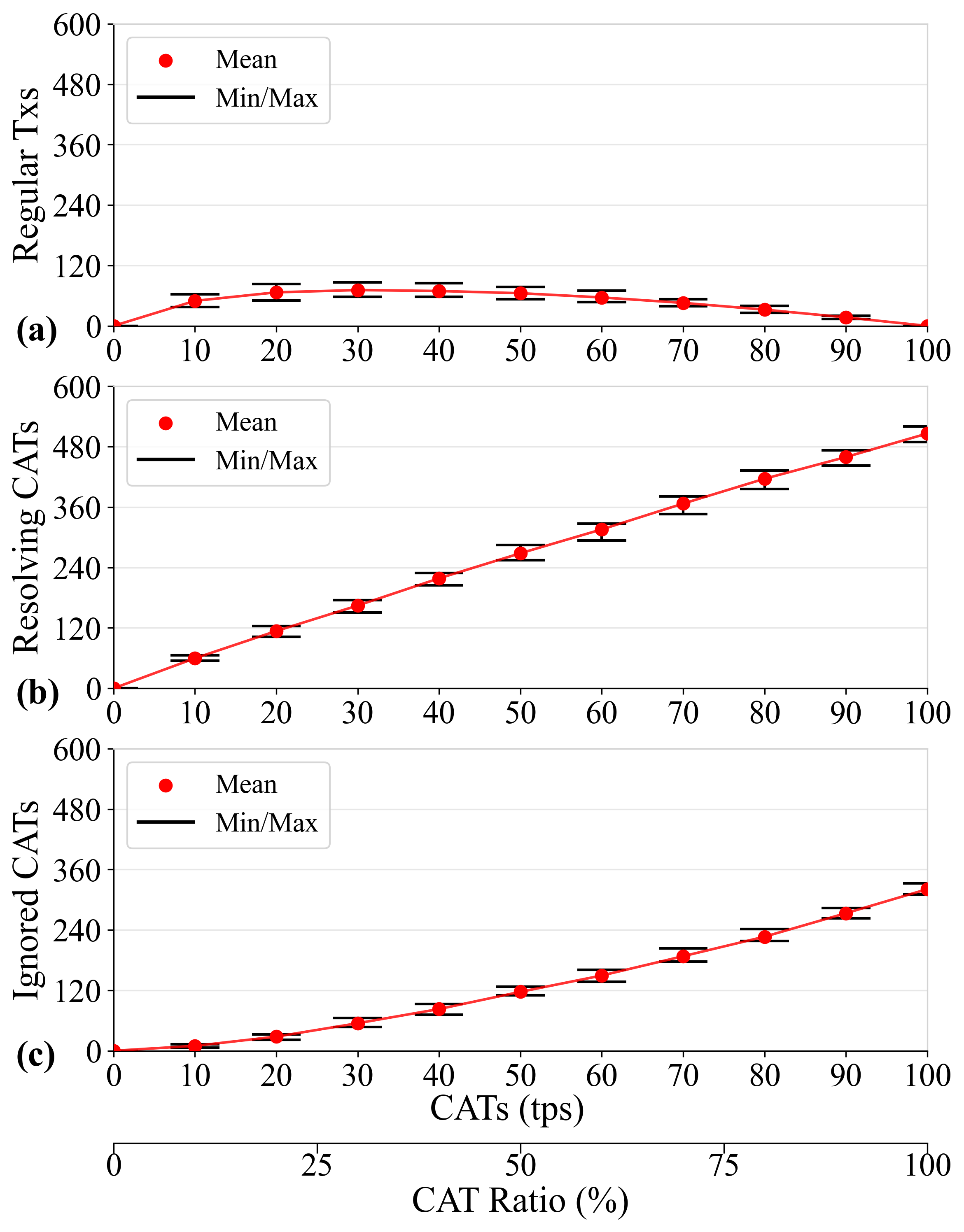}
\caption{Impact of the number of CATs per block on average number of postponed transactions: (a) regular transactions, (b) resolving CATs, and (c) temporarily ignored CATs.}
\label{fig:cat-ratio-pending-combined}
\end{figure}


\subsubsection{CAT lifetime.}
\label{sec:results:cat-lifetime}

As the CAT lifetime increases, CATs have an increased amount of time to get resolved. However, this implies that they are also blocking for an extended amount of time. In turn a large CAT lifetime exposes regular transactions that have a dependency on a pending CAT to increased latency to finality. This can be seen in \figref{fig:cat-lifetime-pending}, where the average time of regular transactions that enter the pool of pending transactions spend in that same pool (i.e., the latency) increases with the CAT lifetime. 

We note that regular transactions that are not dependent on a pending CAT are not affected in their latency, and that transactions that are dependent are maximally delayed by the CAT lifetime. This trade-off is the central design tension of the timeout parameter: the timeout must balance CAT resolution probability against native transaction blocking time. A shorter timeout prioritizes native transaction liveness at the cost of lower CAT success rates, while a longer timeout improves CAT completion but increases worst-case latency for dependent regular transactions, bounded by the lifetime.

\begin{figure}[htbp]
\centering
\includegraphics[width=0.49\textwidth,trim={0 0 0 0cm},clip]{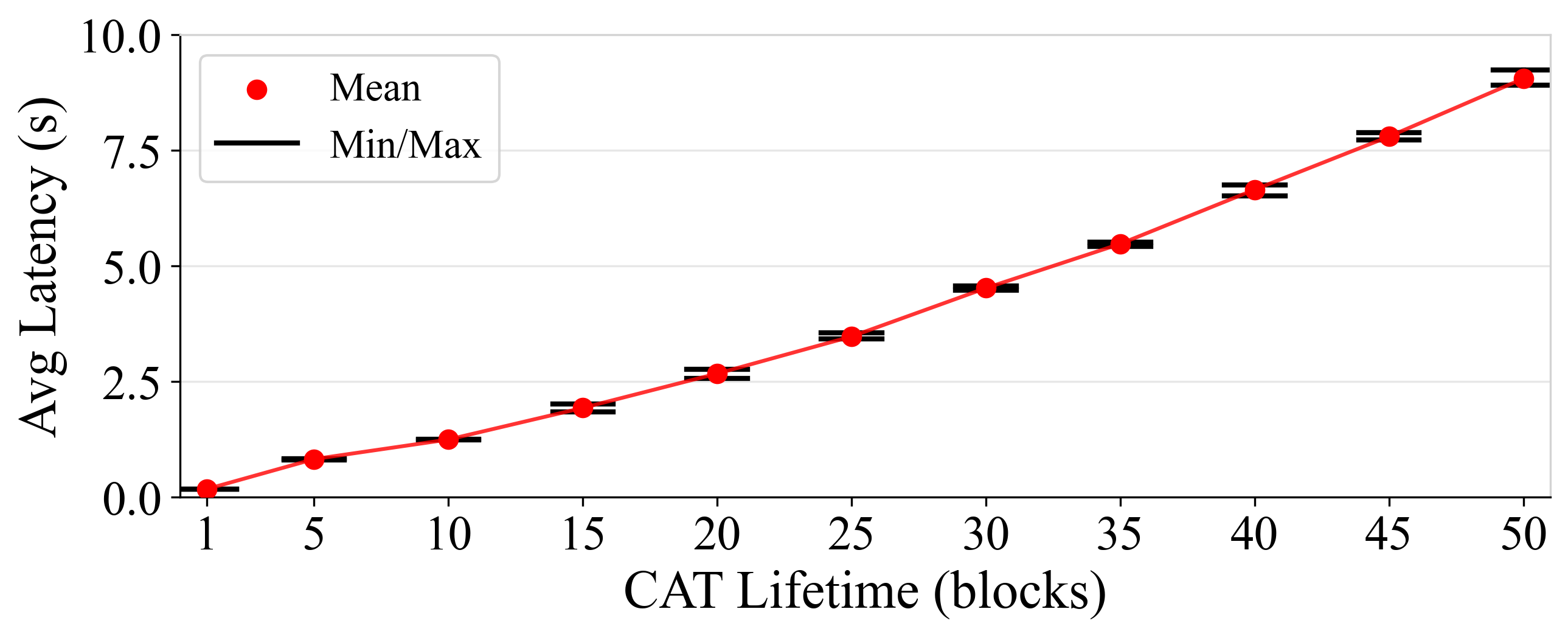}
\caption{Impact of CAT lifetime on the average latency of regular transactions that have a dependency on a CAT.}
\label{fig:cat-lifetime-pending}
\end{figure}


\subsubsection{Chain delay.}
\label{sec:results:chain-delay}

\figref{fig:chain-delay} shows the effect of the chain delay $D$ on the success rate of CATs. The success rate slowly decreases as we increase the chain delay to the CAT lifetime. 
When $D$ is close to the CAT lifetime, mainly CATs are successful that access non-contended keys, while most other CATs would fail due to timeout. This is because a CAT that has to wait for a contended key for $\delta$ blocks would require $D+\delta$ blocks to propose and resolve, which would likely be larger than the CAT lifetime. 
When $D$ is larger than the CAT lifetime, the success rate is 0, as no CATs can be resolved in time.

\begin{figure}[htbp]
\centering
\includegraphics[width=0.49\textwidth,trim={0 0 0 0cm},clip]{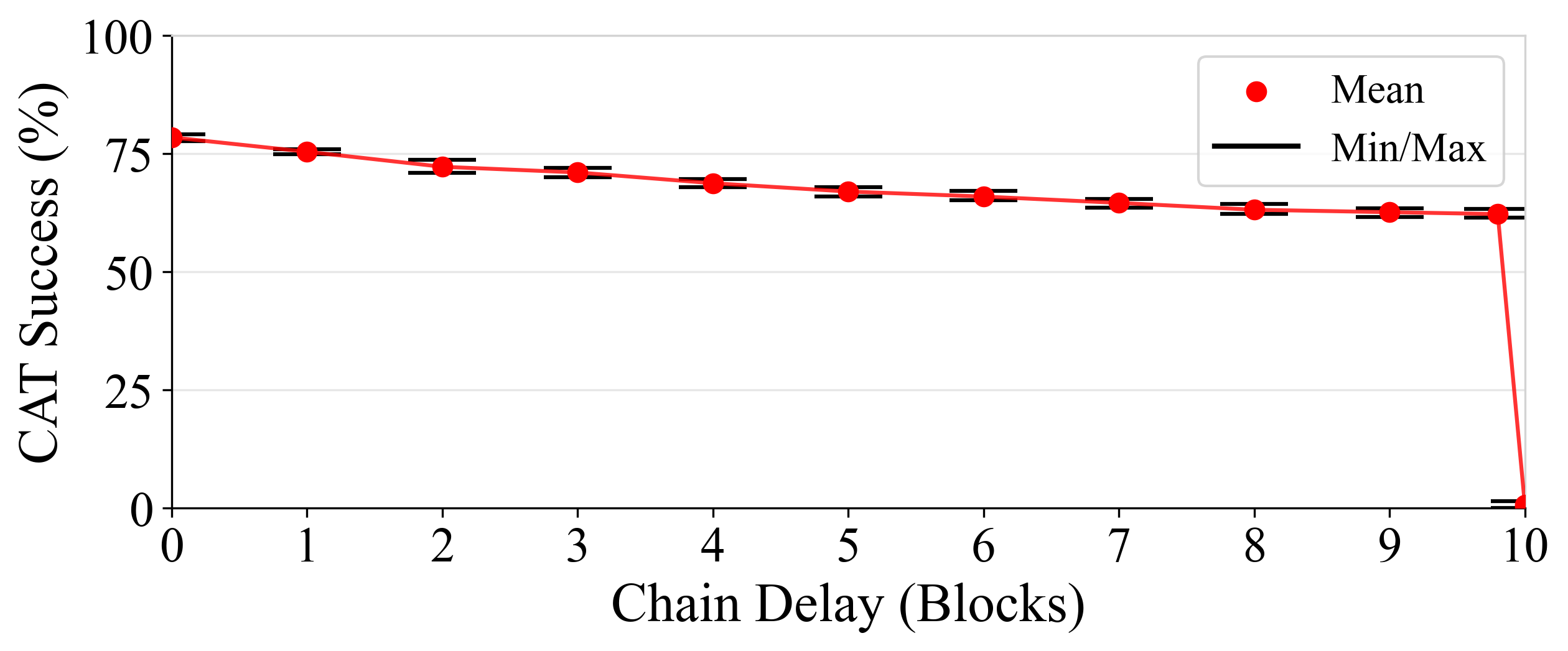}
\vspace{-0.2cm}
\caption{Average success of CATs with the chain delay.}
\label{fig:chain-delay}
\end{figure}


\subsubsection{Number of regular transactions per block.}
\label{sec:results:regular-tx-per-block}

\figref{fig:cat-ratio-v2} shows the success rate of CATs with the number of total transactions per block. We keep the average number of CATs per block constant to the value of 50.
As the throughput increases, the success rate of CATs is decreasing. Since regular transactions that depend on a CAT also block keys temporarily, the footprint of a CAT on locked keys is increasing with the throughput of regular transactions, which is one of the key reasons to keep lifetimes of CATs short.

\begin{figure}[htbp]
\centering
\includegraphics[width=0.49\textwidth,trim={0 0 0 0cm},clip]{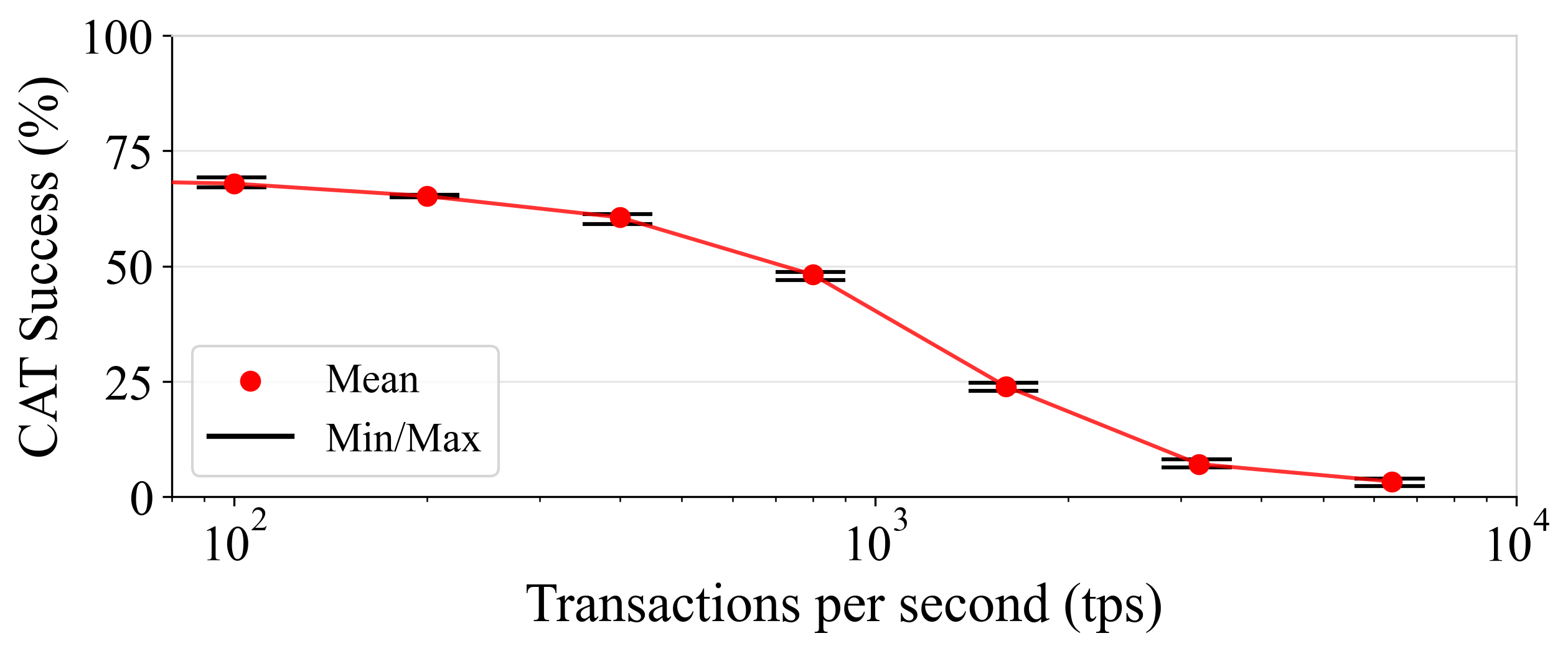}
\caption{Average success of CATs with the throughput, and while keeping the average number of CATs per block constant as per Table~\ref{tab:exp-params}.}
\label{fig:cat-ratio-v2}
\end{figure}


\subsubsection{Centralization of key accesses.}
\label{sec:results:centralization-key-accesses}

\figref{fig:centralization-key-accesses} shows the effect of the centralization of key accesses on the success rate of CATs. We vary the centralization of key accesses by varying the skewness parameter $z$ of the Zipf distribution. The success rate of CATs is high for low centralization, and decreases as the centralization increases, due to increased contention for higher ranked keys (see also the discussion of contention effects in \secref{sec:results:cat-ratio}). We also compare for different values of $r_C$, the ratio of CATs to the total number of transactions. We see the protocol performs particularly well at low centralization of account access, even if $r_C$ is high.

\begin{figure}[htbp]
\centering
\includegraphics[width=0.49\textwidth,trim={0 0 0 1.0cm},clip]{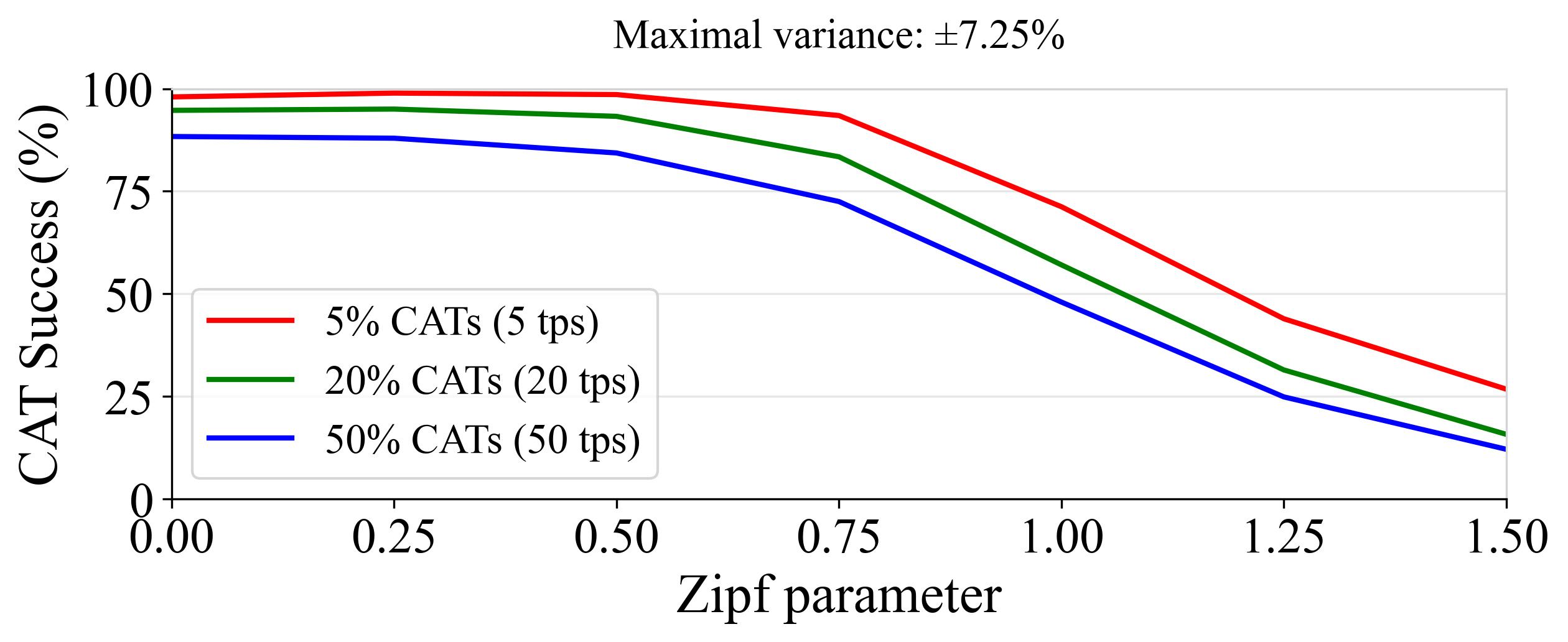}
\caption{Impact of the centralization of key accesses on the average success rate of CATs. We compare for different values of $r_C$, the ratio of CATs to the total number of transactions. Maximal variance of data points is less than 3\%.}
\label{fig:centralization-key-accesses}
\end{figure}

\section{Related Work}
\label{sec:related-work}

In the following, we highlight state-of-the-art blockchain interoperability across several ecosystems.

\subsubsection{Ethereum}

Ethereum is retrofitting interoperability onto a fragmented ecosystem of rollups. Solutions like AggLayer and Espresso aim to coordinate execution and liquidity through shared sequencing and unified bridging, while anchoring security to Ethereum L1. AggLayer provides a shared execution layer between rollups and Ethereum, using zk-proofs to post batched state changes~\cite{espresso2024agglayer,polygon2024misconceptions}.
Espresso~\cite{EspressoSequencer2024} provides a shared sequencing layer using BFT consensus (HotShot) that produces a unified transaction log across participating chains, aiming to enable cross-chain composability~\cite{EspressoCatalyst2023}. The HotShot protocol, including its threat model and data-availability layer (Tiramisu), is described in an ePrint preprint~\cite{EspressoSequencer2024}; however, it focuses on consensus and data availability rather than cross-chain execution semantics. Espresso's CIRC protocol adds inbox/outbox coordination for cross-rollup communication, but to our knowledge has no formal specification or peer-reviewed publication that would permit comparative study. A HotShot-like confirmation layer could serve as the infrastructure underneath our protocol.

\subsubsection{General Message Passing}

General Message Passing (GMP) protocols such as LayerZero~\cite{layerzero2024} provide cross-chain message delivery with verification. However, it does not guarantee atomic execution of interdependent operations across chains: if the destination-side action fails, there is no coordinated abort on the source side. Hash-time-lock contracts (HTLCs) can extend GMP to atomic asset swaps, but are limited to pairwise exchanges and require extensive time-locks. Neither GMP alone nor GMP+HTLC supports atomic execution of general transactions with shared-state dependencies across chains.
Lu~et~al.~\cite{lu2024atomicity} build a two-phase commit (2PC) protocol on top of GMP bridges, achieving atomic execution of general cross-chain transactions. Their protocol uses coarse-grained pessimistic locking: all state variables of affected contracts are locked, blocking even unrelated transactions that access different variables of the same contract. By contrast, our protocol tracks fine-grained read/write dependencies (\secref{sec:dependencies}), postponing only truly dependent transactions. Moreover, no timeout or fallback mechanism is specified if the proposer becomes unresponsive, leaving liveness without guarantees.

\subsubsection{Polkadot}

Polkadot provides interoperability via its Relay Chain, which coordinates shared consensus and security across parachains~\cite{burdges2020polkadot,polkadot2023overview}. Cross-chain communication uses XCM, an asynchronous message format that follows a fire-and-forget model~\cite{wood2021xcm}: the sender does not block on completion, and if destination-side execution fails, there is no automatic rollback on the source chain. Error-handling instructions operate only on the destination chain and cannot trigger a coordinated abort across chains. Thus, despite being a multi-chain system with native cross-chain communication, Polkadot does not provide atomic cross-chain execution.

\subsubsection{Cosmos}

Cosmos uses the IBC protocol for asynchronous, trust-minimized messaging across sovereign chains~\cite{chervinski2023ibc}. IBC provides message passing but not atomic execution; each chain maintains its own validators and consensus.
Avalon~\cite{Cai2024cats-through-state-layers} builds on IBC to achieve complete atomicity across Cosmos zones, introducing a \emph{dirty state layer} for caching state changes before commit and applying optimistic concurrency control (OCC) with a state synchronization protocol.
While similar in intent, Avalon's dirty queue commits transactions in strict order: all subsequent transactions in the dirty queue are aborted, even if they access disjoint state. By contrast, our protocol builds an explicit dependency graph that enables independent transactions to proceed in parallel past pending CATs, and bounds cascading dependencies via the $\texttt{maxDepth}$ parameter (\secref{sec:model}). We compare latency and scaling characteristics in \secref{sec:results}.

\subsubsection{Relation to Cross-chain Atomicity and Sharding}

Herlihy~\cite{herlihy2019cross} formalizes the cross-chain atomicity problem and shows that under semi-synchronous communication - a realistic assumption for blockchain systems - a globally shared ledger is required for coordination; our confirmation layer is precisely such a ledger.
The same coordination challenge arises in sharded blockchains, where independent partitions must agree on joint commit-or-abort decisions.
CAPER~\cite{amiri2019caper} addresses cross-application transactions in permissioned blockchains but conflates ordering and execution trust in permissioned chain nodes; we separate them - execution correctness is enforced via ZK proofs or staking on the transaction processor, while ordering is the confirmation layer's exclusive responsibility.
ByShard~\cite{hellings2023byshard} validates the coordinator-then-local-execution pattern in a Byzantine cross-shard setting: their Orchestrate-Execute Model applies a similar structure - a coordinator decides commit/abort atomically, then each shard executes locally.
Our protocol unifies these insights: it applies the shared-ledger requirement established by Herlihy via the confirmation layer, while adopting the coordinator-then-execute pattern validated by ByShard, in a permissionless cross-chain setting.

\subsubsection{Comparison of Cross-chain Approaches}

Table~\ref{tab:comparison} compares protocols that target atomic cross-chain execution and have published protocol details. Polkadot and Cosmos/IBC are excluded as they provide message passing but not atomic execution; we could not identify peer-reviewed protocol specifications or performance studies for AggLayer or Espresso/CIRC (\cite{EspressoSequencer2024} covers consensus and DA but not cross-chain execution). GMP+HTLC denotes GMP combined with hash-time-lock contracts; GMP+2PC denotes the two-phase commit protocol of Lu~et~al.

\begin{table}[htbp]
\centering
\renewcommand{\arraystretch}{1.3}
\footnotesize
\resizebox{\columnwidth}{!}{%
\begin{tabular}{@{}lcccc@{}}
\toprule
& \textbf{CATs} & \textbf{Avalon} & \textbf{GMP+2PC} & \textbf{GMP+HTLC} \\
& (ours) & \cite{Cai2024cats-through-state-layers} & \cite{lu2024atomicity} & \\
\midrule
Generic atomic exec. & \checkmark & \checkmark & \checkmark & Swaps \\
Dependency tracking\footnotemark & \checkmark & $\times$ & $\times$ & $\times$ \\
Liveness guarantee & \checkmark & \checkmark & $\times$ & Time-locks \\
Comm.\ complexity & $O(n)$ & $O(n^2)$ & $O(n)$ & Pairwise \\
\multirow{2}{*}{Reported latency} & \multirow{2}{*}{$\sim$1-2s} & 15s & $\sim$300s & \multirow{2}{*}{Time-locks} \\
 & & (3 chains) & (2 chains) & \\
\bottomrule
\end{tabular}%
}
\vspace{4pt}
\caption{Comparison of cross-chain coordination approaches. Communication complexity is given in the number of participating chains $n$.}
\label{tab:comparison}
\end{table}
\footnotetext{Explicit tracking of read/write dependencies between cross-chain and regular transactions (\secref{sec:dependencies}). This allows independent transactions to proceed in parallel past pending CATs, while only truly dependent transactions are postponed. See the Avalon comparison above for the contrasting approach.}

\section{Conclusion}
\label{sec:conclusion}

We have presented a novel protocol for cross-chain atomic transactions that addresses the fundamental challenges of blockchain interoperability: fragmented liquidity, high latency, and lack of atomicity. Our approach introduces a shared coordination layer with executors, coordinators, and a confirmation layer that enables deterministic, secure cross-chain execution while preserving chain autonomy.

The protocol's key innovations include simulation before commitment to avoid rollback complexity, dependency tracking to enable independent transaction execution, and timeout mechanisms to ensure liveness under asynchronous conditions. Through our implementation in Hyperplane, we demonstrate that the protocol efficiently achieves atomicity without centralized coordination or speculative rollback.

Our experimental evaluation reveals critical performance characteristics and trade-offs.
The protocol maintains robust performance as the proportion of CATs increases, provided that access pattern centralization remains moderate.
This resilience to transaction pattern variations makes the protocol suitable for diverse real-world applications.
We find that CAT success rates are sensitive to chain delays when delays approach the CAT lifetime threshold. On the other hand, while increasing CAT lifetime would improve success rates, it would increase pending regular transaction backlogs, which in turn extends finality times for dependent regular transactions, highlighting the importance of careful parameter tuning for specific use cases.

Our analysis establishes strong safety and liveness guarantees, showing that the protocol provides Byzantine fault tolerance while maintaining minimal blocking through dependency-aware scheduling. The protocol provides a foundation for truly composable cross-chain applications, enabling the seamless multi-chain experience that users expect while maintaining the security and decentralization properties of blockchain systems.

\paragraph{Future work.} Several directions remain for future investigation. First, while $\texttt{maxDepth}=1$ is the motivated operating point, an experimental sweep over varying dependency-depth bounds would further characterize the trade-off between blocking reduction and CAT success rate. Second, extending the evaluation to more than two chains would empirically confirm the protocol's constant-round scaling property. Third, evaluating the protocol under standardized database benchmarks such as TPC-C would enable direct comparison with the cross-shard literature.

\appendix
\LinesNotNumbered
\section{Algorithms}
\label{sec:algorithms}

This appendix provides detailed algorithmic descriptions of the key components of our CAT protocol. These algorithms implement the theoretical framework described in Section~\ref{sec:protocol} of the main text and provide concrete implementations for the resolver, executor, and coordinator components.

The algorithms are designed to work together in a coordinated manner, where each component has specific responsibilities:
\begin{itemize}
    \item \textbf{Resolver}: Processes pending transactions and resolves their status based on coordinator decisions
    \item \textbf{Executor}: Manages transaction execution, reordering, and dependency resolution
    \item \textbf{Coordinator}: Orchestrates cross-chain coordination and resolves CAT outcomes
\end{itemize}

\subsection{Resolver Algorithm}
\label{sec:algo:resolver}

The resolver is responsible for processing the pending transaction set from previous rounds and updating transaction statuses based on coordinator decisions, recorded at the Confirmation Layer.

This algorithm implements the resolver component described in Section~\ref{sec:protocol} of the main text. Algorithm~\ref{alg:resolver} shows the complete resolver implementation.

\begin{algorithm}[t]
  \caption{Resolver Algorithm}
  \label{alg:resolver}
  \KwIn{Resolved statuses $\status^{r}$ and pending sequence $\sigma_P^{r-1}$, where $r$ is the current round}
  \KwOut{Resolved transaction sequence $\tilde{\sigma}^r$}

  \vspace{1em}
  Initialize empty sequence $\tilde{\sigma}^r\gets [\,]$\;

  \For{each transaction $t$ in $\sigma_P^{r-1}$}{
    \uIf{$t$ is not a CAT}{
      Append $t$ to $\tilde{\sigma}^r$\;
    }
    \Else{
      Let $T$ be the CAT identifier in $t$\;
      \uIf{$\status^{r}[T] = \success$}{
          Append $t$ to $\tilde{\sigma}^r$\tcp*{the CAT becomes a regular transaction}
      }
      \uElseIf{$\status^{r}[T] = \failure$}{
        Append $\Skip$ to $\tilde{\sigma}^r$\;
      }
      \Else{
        Append $t$ to $\tilde{\sigma}^r$ \emph{(with flag pending)}\;
      }
    }
  }

  \Return $\tilde{\sigma}^r$\;
\end{algorithm}

\subsection{Executor with Timeouts Algorithm}
\label{sec:algo:executor-with-timeouts}

This version of the executor introduces timeout mechanisms to prevent indefinite blocking and ensure protocol liveness. The timeout mechanism is crucial for maintaining system responsiveness while preserving the safety guarantees of the CAT protocol.

This algorithm implements the timeout mechanism described in Section~\ref{sec:cat-protocol:timeouts} of the main text. Algorithm~\ref{alg:executor-with-timeouts} shows the complete executor with timeouts implementation. The timeout-related additions are marked in red.

\begin{algorithm}[t]
  \caption{Executor Algorithm with Timeout}
  \label{alg:executor-with-timeouts}
  \KwIn{State $s^{r-1}$,\\
        \hspace{3.1em} Transaction sequence $\sigma = \tilde{\sigma}^r . \sigma^r$ with new CATs in $\sigma^r$ marked \pending,\\
        \hspace{3.1em} \red{timeout threshold $\Delta$ and current round $r$},\\
        \hspace{3.1em} \red{map $\mathsf{timestamp}(T)$ which provides the received round for CAT $T$}}
  \KwOut{Accepted sequence $\sigma_A^r$, postponed sequence $\sigma_P^r$, proposal map $\propose$, updated state $s^r$}

  \vspace{1em}
  Initialize empty sequences $\sigma_A^r \gets [\,]$, $\sigma_P^r \gets [\,]$, and empty map $\propose \gets \emptymap$\;

  \For{each transaction $t$ in $\sigma$, in order}{
    Compute local outcome $\mathsf{result} \in \{\success, \failure\}$\;
    
    \uIf{$t$ is regular}{
      \uIf{$t \independent{s^{r-1}} \sigma_P^r[i]$ for all $i$}{
        Append $t$ to $\sigma_A^r$\;
      }
      \Else{
        Append $t$ to $\sigma_P^r$\;
      }
    }
    \ElseIf{$t$ is a \pending CAT}{
      Let $T$ be the CAT identifier in $t$\;
      \uIf{\red{$\mathsf{timestamp}(T) \le r - \Delta$}}{
        \red{Append $\Skip$ to $\sigma_A^r$} \tcp*{\red{forced resolution via timeout}}
      }
      \Else{
        \If{$T$ has not been proposed in any earlier round}{
          Insert $(T, \mathsf{result})$ into $\propose$\;
        }
        Append $t$ to $\sigma_P^r$\;
      }
    }
  }

  Compute updated state: $s^r \gets \Next(s^{r-1}, \sigma_A^r)$\;

  \Return $\sigma_A^r$, $\sigma_P^r$, $\propose$, $s^r$\;
\end{algorithm}

\subsection{Executor with Dependency Depth Limiting Algorithm}
\label{sec:algo:executor-v2}

This version of the executor introduces dependency depth limiting to prevent resource exhaustion attacks and improve system scalability. This algorithm implements dependency management strategies that balance performance with security considerations.

This algorithm implements the dependency depth limiting approach described in Section~\ref{sec:cat-protocol:limiting-dependency-depth} of the main text. Algorithm~\ref{alg:executor-depth-limited} shows the complete executor with dependency depth limiting implementation. The depth-limiting additions are marked in blue.

\begin{algorithm}[t]
  \caption{Executor Algorithm with Dependency Depth Limit}
  \label{alg:executor-depth-limited}
  \KwIn{State $s^{r-1}$,
  Transaction sequence \blue{$\sigma = \tilde{\sigma}^r .\calI^{r-1} . \sigma^r$} with new CATs in $\sigma^r$ marked \pending,\\
  \hspace{3.1em} \blue{and maximum depth $\texttt{maxDepth}$}}
  \KwOut{Accepted sequence $\sigma_A^r$, postponed sequence $\sigma_P^r$, \blue{ignored sequence $\calI^{r}$}, proposal map $\propose$, updated state $s^r$}

  \vspace{1em}
  Initialize empty sequences $\sigma_A^r \gets [\,]$, $\sigma_P^r \gets [\,]$, \blue{$\calI^{r} \gets [\,]$}, and empty map $\propose \gets \emptymap$\;

  \For{each transaction $t$ in $\sigma$, in order}{
    Compute local outcome $\mathsf{result} \in \{\success, \failure\}$\;
    \blue{Compute $\depth_{s^{r-1}, \sigma_P^r}(t)$}\;

    \uIf{\blue{$t$ is a \cat{} and $\depth_{s^{r-1}, \sigma_P^r}(t) > \texttt{maxDepth}$}}{
      \blue{Add $t$ to the ignored set $\calI^{r}$}\;
    }
    \uElseIf{\blue{$t$ is not a \cat{} and $\depth_{s^{r-1}, \sigma_P^r}(t) > \texttt{maxDepth}$}}{
      \blue{Add $t$ to the ignored set $\calI^{r}$}\;
    }
    \uElseIf{$t$ is regular}{
      \uIf{$t \independent{s^{r-1}} \sigma_P^r[i]$ for all $i$}{
        Append $t$ to $\sigma_A^r$\;
      }
      \Else{
        Append $t$ to $\sigma_P^r$\;
      }
    }
    \ElseIf{$t$ is a \pending CAT}{
      Let $T$ be the CAT identifier in $t$\;
      \If{$T$ has not been proposed in any earlier round}{
        Insert $(T, \mathsf{result})$ into $\propose$\;
      }
      Append $t$ to $\sigma_P^r$\;
    }
  }

  Compute updated state: $s^r \gets \Next(s^{r-1}, \sigma_A^r)$\;

  \Return $\sigma_A^r$, $\sigma_P^r$, $\propose$, \blue{$\calI^{r}$}, $s^r$\;
\end{algorithm}

\subsection{Coordinator Algorithm}
\label{sec:algo:coordinator}

The coordinator serves as the component where proposals are aggregated from the executor of all participating chains and which forwards a status to the Confirmation Layer once all required proposals are received.

This algorithm implements the coordinator component described in Section~\ref{sec:coordinator} of the main text. Algorithm~\ref{alg:coordinator} shows the complete coordinator implementation.

\begin{algorithm}[t]
  \caption{Event-Driven CAT Status Resolution by Coordinator}
  \label{alg:coordinator}
  \KwIn{Proposal $\propose_c^r$ from chain $c$ for round $r$,\\
  \hspace{3.2em} which is a partial map: $T \mapsto \status$ for some CATs $T$}
  
  \vspace{1em}
  Initialize empty queue $Q \gets [\,]$\;
  
  \For{each CAT $T \in \mathrm{domain}(\propose_c^r)$}{
    Store proposed $\status$ for $T$ from chain $c$\;
    \If{proposals from all chains for $T$ are received}{
      \uIf{all chains propose \textsf{success}}{
        $\status^r[T] \gets \success$\;
      }
      \uElse{
        $\status^r[T] \gets \failure$\;
      }
      Insert $(T, \status^r[T])$ into $Q$\;
    }
  }
  \If{$Q$ is not empty}{
    Send all entries in $Q$ to the Confirmation Layer (CL)\;
  }
\end{algorithm}

\bibliographystyle{IEEEtran}
\bibliography{references}

@inproceedings{Zamyatin2021,
  author    = {Zamyatin, Alexei
               and Al-Bassam, Mustafa
               and Zindros, Dionysis
               and Kokoris-Kogias, Eleftherios
               and Moreno-Sanchez, Pedro
               and Kiayias, Aggelos
               and Knottenbelt, William J.},
  editor    = {Borisov, Nikita
               and Diaz, Claudia},
  title     = {SoK: Communication Across Distributed Ledgers},
  booktitle = {Financial Cryptography and Data Security},
  year      = {2021},
  publisher = {Springer Berlin Heidelberg},
  address   = {Berlin, Heidelberg},
  pages     = {3--36},
  isbn      = {978-3-662-64331-0}
}

@inproceedings{Gelashvili2023,
  author    = {Gelashvili, Rati and Spiegelman, Alexander and Xiang, Zhuolun and Danezis, George and Li, Zekun and Malkhi, Dahlia and Xia, Yu and Zhou, Runtian},
  title     = {Block-STM: Scaling Blockchain Execution by Turning Ordering Curse to a Performance Blessing},
  year      = {2023},
  isbn      = {9798400700156},
  publisher = {Association for Computing Machinery},
  address   = {New York, NY, USA},
  url       = {https://doi.org/10.1145/3572848.3577524},
  doi       = {10.1145/3572848.3577524},
  booktitle = {Proceedings of the 28th ACM SIGPLAN Annual Symposium on Principles and Practice of Parallel Programming},
  pages     = {232--244},
  numpages  = {13},
  location  = {Montreal, QC, Canada},
  series    = {PPoPP '23}
}

@misc{polygon2024misconceptions,
  author       = {Polygon Labs},
  title        = {Clearing Up Agglayer Misconceptions},
  year         = {2024},
  month        = nov,
  url          = {https://polygon.technology/blog/clearing-up-agglayer-misconceptions},
  note         = {Accessed: 2025-07-24}
}

@misc{espresso2024agglayer,
  author       = {EspressoSystems},
  title        = {Espresso is solving rollup interoperability with the AggLayer and Polygon Labs},
  year         = {2024},
  month        = may,
  howpublished = {https://medium.com/@espressosys/espresso-is-solving-rollup-interoperability-with-the-agglayer-and-polygon-labs-b3a7d2f8f7cf},
  note         = {Accessed: 2025-07-24}
}

@misc{agglayer2025,
  author       = {{AggLayer}},
  title        = {AggLayer Documentation},
  year         = {2025},
  url          = {https://docs.agglayer.dev/},
  note         = {Accessed: 2025-09-02}
}

@misc{Cai2024cats-through-state-layers,
      author = {Yuandi Cai and Ru Cheng and Yifan Zhou and Shijie Zhang and Jiang Xiao and Hai Jin},
      title = {Enabling Complete Atomicity for Cross-chain Applications Through Layered State Commitments},
      howpublished = {Cryptology {ePrint} Archive, Paper 2024/1084},
      year = {2024},
      url = {https://eprint.iacr.org/2024/1084}
}

@article{zakhary2020atomic,
  title     = {Atomic Commitment Across Blockchains},
  author    = {Zakhary, Victor and Agrawal, Divyakant and El Abbadi, Amr},
  journal   = {Proceedings of the VLDB Endowment},
  volume    = {13},
  number    = {9},
  pages     = {1319--1331},
  year      = {2020},
  publisher = {VLDB Endowment},
  url       = {https://www.vldb.org/pvldb/vol13/p1319-zakhary.pdf}
}

@misc{EspressoCatalyst2023,
  title = {Espresso Systems and Catalyst Collaborate to Improve Interoperability},
  author = {{Espresso Systems}},
  year = {2023},
  howpublished = {https://medium.com/@espressosys/espresso-systems-and-catalyst-collaborate-to-improve-interoperability-239addbe2c2b},
  note = {Accessed 2025-07-24}
}

@misc{EspressoSequencer2024,
  title = {The Espresso Sequencing Network: {HotShot} Consensus, {Tiramisu} Data-Availability, and Builder-Exchange},
  author = {Jeb Bearer and Benedikt B{\"u}nz and Philippe Camacho and Binyi Chen and Ellie Davidson and Ben Fisch and Brendon Fish and Gus Gutoski and Fernando Krell and Chengyu Lin and Dahlia Malkhi and Kartik Nayak and Keyao Shen and Alex Xiong and Nathan Yospe and Sishan Long},
  year = {2024},
  howpublished = {Cryptology ePrint Archive, Paper 2024/1189},
  url = {https://eprint.iacr.org/2024/1189},
  note = {Preprint}
}

@article{lu2024atomicity,
  title   = {Atomicity and Abstraction for Cross-Blockchain Interactions},
  author  = {Lu, Huaixi and Jajoo, Akshay and Namjoshi, Kedar S.},
  journal = {arXiv preprint arXiv:2403.07248},
  year    = {2024},
  url     = {https://arxiv.org/abs/2403.07248}
}

@misc{polkadot2023overview,
  author       = {{Polkadot Developers}},
  title        = {Overview: Polkadot Chain Architecture},
  year         = {2023},
  url          = {https://docs.polkadot.com/polkadot-protocol/architecture/polkadot-chain/overview},
  note         = {Accessed: 2025-07-24}
}

@article{burdges2020polkadot,
  author       = {Burdges, Jeff and Cevallos, Alfonso and Czaban, Peter and Habermeier, Rob and Hosseini, Syed and Lama, Fabio and Alper, Handan Kilinc and Luo, Ximin and Shirazi, Fatemeh and Stewart, Alistair and Wood, Gavin},
  title        = {Overview of {Polkadot} and its Design Considerations},
  journal      = {arXiv preprint arXiv:2005.13456},
  year         = {2020}
}

@misc{wood2021xcm,
  author       = {Wood, Gavin},
  title        = {{XCM} Part {III}: Execution and Error Management},
  year         = {2021},
  howpublished = {Polkadot Blog},
  url          = {https://medium.com/polkadot-network/xcm-part-iii-execution-and-error-management-ceb8155dd166},
  note         = {Accessed: 2026-03-08}
}

@misc{erc7683,
  title        = {ERC-7683: Standard for Crosschain Intents},
  howpublished = {\url{https://www.erc7683.org/}},
  note         = {Accessed: 2025-07-25},
  year         = {2024},
  author       = {{ERC-7683 Contributors}},
}

@misc{gelashvili2022blockstmscalingblockchainexecution,
      title={Block-STM: Scaling Blockchain Execution by Turning Ordering Curse to a Performance Blessing}, 
      author={Rati Gelashvili and Alexander Spiegelman and Zhuolun Xiang and George Danezis and Zekun Li and Dahlia Malkhi and Yu Xia and Runtian Zhou},
      year={2022},
      eprint={2203.06871},
      archivePrefix={arXiv},
      primaryClass={cs.DC},
      url={https://arxiv.org/abs/2203.06871}, 
}

@misc{müller2023realitybasedutxoledger,
      title={Reality-based UTXO Ledger}, 
      author={Sebastian Müller and Andreas Penzkofer and Nikita Polyanskii and Jonas Theis and William Sanders and Hans Moog},
      year={2023},
      eprint={2205.01345},
      archivePrefix={arXiv},
      primaryClass={cs.DC},
      url={https://arxiv.org/abs/2205.01345}, 
}

@article{kung1981optimistic,
  author    = {H. T. Kung and John T. Robinson},
  title     = {On Optimistic Methods for Concurrency Control},
  journal   = {ACM Transactions on Database Systems (TODS)},
  volume    = {6},
  number    = {2},
  pages     = {213--226},
  year      = {1981},
  publisher = {ACM},
  doi       = {10.1145/319566.319567}
}

@article{lamport1998parttime,
author = {Lamport, Leslie},
title = {The part-time parliament},
year = {1998},
issue_date = {May 1998},
publisher = {Association for Computing Machinery},
address = {New York, NY, USA},
volume = {16},
number = {2},
issn = {0734-2071},
url = {https://doi.org/10.1145/279227.279229},
doi = {10.1145/279227.279229},
journal = {ACM Trans. Comput. Syst.},
month = may,
pages = {133–169},
numpages = {37}
}

@inproceedings{castro1999practical,
author = {Castro, Miguel and Liskov, Barbara},
title = {Practical Byzantine fault tolerance},
year = {1999},
isbn = {1880446391},
publisher = {USENIX Association},
address = {USA},
booktitle = {Proceedings of the Third Symposium on Operating Systems Design and Implementation},
pages = {173–186},
numpages = {14},
location = {New Orleans, Louisiana, USA},
series = {OSDI '99}
}

@article{dwork1988consensus,
author = {Dwork, Cynthia and Lynch, Nancy and Stockmeyer, Larry},
title = {Consensus in the presence of partial synchrony},
year = {1988},
issue_date = {April 1988},
publisher = {Association for Computing Machinery},
address = {New York, NY, USA},
volume = {35},
number = {2},
issn = {0004-5411},
url = {https://doi.org/10.1145/42282.42283},
doi = {10.1145/42282.42283},
journal = {J. ACM},
month = apr,
pages = {288–323},
numpages = {36}
}

@article{kusmierz2021,
  author  = {Bartosz Ku{\'s}mierz and Sebastian M{\"u}ller and Angelo Capossele},
  title   = {Committee Selection in DAG Distributed Ledgers and Applications},
  journal = {Intelligent Computing},
  year    = {2021}
}

@misc{tonkikh2025raptrprefixconsensusrobust,
      title={Raptr: Prefix Consensus for Robust High-Performance BFT}, 
      author={Andrei Tonkikh and Balaji Arun and Zhuolun Xiang and Zekun Li and Alexander Spiegelman},
      year={2025},
      eprint={2504.18649},
      archivePrefix={arXiv},
      primaryClass={cs.DC},
      url={https://arxiv.org/abs/2504.18649}, 
}

@misc{Robinson2020Performance,
  author    = {Peter Robinson},
  title     = {{Performance Overhead of Atomic Crosschain Transactions}},
  year      = {2020},
  howpublished = {arXiv preprint arXiv:2005.10684},
  note      = {Published 19 May 2020},
  url       = {https://arxiv.org/abs/2005.10684}
}

@misc{RobinsonRamesh2020Layer2,
  author    = {Peter Robinson and Raghavendra Ramesh},
  title     = {{Layer 2 Atomic Cross-Blockchain Function Calls}},
  year      = {2020},
  howpublished = {arXiv preprint arXiv:2005.09790},
  note      = {Published 19 May 2020},
  url       = {https://arxiv.org/abs/2005.09790}
}

@misc{Paradigm2023Intents,
  author       = {Paradigm},
  title        = {{Intent-Based Architecture and Their Risks}},
  year         = {2023},
  howpublished = {Paradigm blog},
  note         = {Accessed: 2025-08-14},
  url          = {https://www.paradigm.xyz/2023/06/intents}
}

@misc{AnomaIntents2022,
  author       = {Anoma Foundation},
  title        = {{Intents and the Intent Gossip Network}},
  year         = {2022},
  howpublished = {\url{https://anoma.net/blog/intents-and-intent-gossip-network}},
  note         = {Accessed: 2025-08-14}
}

@article{herlihy2019cross,
  author    = {Maurice Herlihy and Barbara Liskov and Liuba Shrira},
  title     = {Cross-chain Deals and Adversarial Commerce},
  journal   = {Proceedings of the VLDB Endowment},
  volume    = {13},
  number    = {2},
  pages     = {100--113},
  year      = {2019},
  publisher = {VLDB Endowment},
  doi       = {10.14778/3364324.3364326},
  url       = {https://www.vldb.org/pvldb/vol13/p100-herlihy.pdf}
}

@inproceedings{amiri2019caper,
  author    = {Mohammad Javad Amiri and Divyakant Agrawal and Amr El Abbadi},
  title     = {{CAPER}: A Cross-Application Permissioned Blockchain},
  booktitle = {Proceedings of the VLDB Endowment},
  volume    = {12},
  number    = {11},
  pages     = {1385--1398},
  year      = {2019},
  publisher = {VLDB Endowment},
  doi       = {10.14778/3342263.3342275},
  url       = {https://www.vldb.org/pvldb/vol12/p1385-amiri.pdf}
}

@article{hellings2023byshard,
  author    = {Jelle Hellings and Mohammad Sadoghi},
  title     = {{ByShard}: Sharding in a Byzantine Environment},
  journal   = {The VLDB Journal},
  volume    = {32},
  pages     = {1343--1367},
  year      = {2023},
  publisher = {Springer},
  doi       = {10.1007/s00778-023-00794-0},
  url       = {https://link.springer.com/article/10.1007/s00778-023-00794-0},
  note      = {Conference version: \url{https://vldb.org/pvldb/vol14/p2230-hellings.pdf}}
}

@inproceedings{chervinski2023ibc,
  author    = {Jo{\~a}o Otavio Chervinski and Diego Kreutz and Jin Yu},
  title     = {Analyzing the Performance of the Inter-Blockchain Communication Protocol},
  booktitle = {Proceedings of the 53rd Annual IEEE/IFIP International Conference on Dependable Systems and Networks (DSN)},
  year      = {2023},
  pages     = {53--65},
  publisher = {IEEE},
  doi       = {10.1109/DSN58367.2023.00021},
  url       = {https://arxiv.org/abs/2303.10844}
}

@article{schneider1990smr,
  author    = {Fred B. Schneider},
  title     = {Implementing Fault-Tolerant Services Using the State Machine Approach: A Tutorial},
  journal   = {ACM Computing Surveys},
  volume    = {22},
  number    = {4},
  pages     = {299--319},
  year      = {1990},
  publisher = {ACM},
  doi       = {10.1145/98163.98167},
  url       = {https://www.cs.cornell.edu/fbs/publications/SMSurvey.pdf}
}

@article{guerraoui2002nbac,
  author    = {Rachid Guerraoui},
  title     = {Non-Blocking Atomic Commit in Asynchronous Distributed Systems with Failure Detectors},
  journal   = {Distributed Computing},
  volume    = {15},
  number    = {1},
  pages     = {17--25},
  year      = {2002},
  publisher = {Springer},
  doi       = {10.1007/s00446-002-8027-4},
  url       = {https://link.springer.com/article/10.1007/s00446-002-8027-4}
}

@article{zheng2025solana,
  author    = {Xiaoye Zheng and Zhiyuan Wan and David Lo and Difan Xie and Xiaohu Yang},
  title     = {Why Does My Transaction Fail? {A} First Look at Failed Transactions on the {Solana} Blockchain},
  journal   = {Proceedings of the ACM on Software Engineering},
  volume    = {2},
  number    = {ISSTA},
  year      = {2025},
  publisher = {ACM},
  doi       = {10.1145/3728943},
  url       = {https://dl.acm.org/doi/abs/10.1145/3728943}
}

@misc{layerzero2024,
  author       = {{LayerZero Labs}},
  title        = {{LayerZero}: An Omnichain Interoperability Protocol},
  year         = {2024},
  howpublished = {\url{https://layerzero.network/}},
  note         = {Accessed: 2025-07-24}
}

@techreport{tpcc2010,
  title        = {{TPC} Benchmark {C}: Standard Specification, Revision 5.11},
  author       = {{Transaction Processing Performance Council}},
  year         = {2010},
  institution  = {TPC},
  url          = {https://www.tpc.org/tpc_documents_current_versions/pdf/tpc-c_v5.11.0.pdf},
}

\end{document}